\documentclass[fleqn]{tlp}

\usepackage{latexsym}
\usepackage{amssymb}
\usepackage{amsmath}
\usepackage[all]{xy}

\newtheorem{definition}{Definition} 
\newtheorem{example}{Example} 
\newtheorem{lemma}{Lemma} 
\newtheorem{proposition}{Proposition} 
\newtheorem{theorem}{Theorem} 

\usepackage{color}

\makeatletter
\newenvironment{prog}{\vspace{0.7ex}\par
\setlength{\parindent}{0.7cm}
\obeylines\@vobeyspaces\tt}{\vspace{0.7ex}\noindent
}
\makeatother
\newcommand{\startprog}{\begin{prog}}
\newcommand{\stopprog}{\end{prog}\noindent}
\newcommand{\conc}{\mathsf{conc}} 
\newcommand{\symb}{\mathsf{symb}} 
\newcommand{\alt}{\mathsf{alt}} 
\newcommand{\traces}{\mathit{Traces}} 
\newcommand{\depth}{\mathit{depth}} 
\newcommand{\rt}{\mathit{root}} 
\newcommand{\inp}{\mathsf{input}} 
\newcommand{\trace}{\mathit{trace}}

\newcommand{\id}{{\mathit{id}}} 

\newcommand{\sgeq}{\geqslant}
\newcommand{\sleq}{\leqslant}

\newcommand{\alttrace}{\mathsf{alt\_trace}}
\newcommand{\true}{\mathsf{true}}
\newcommand{\mgu}{\mathsf{mgu}}
\newcommand{\error}{\mathsf{error}}
\newcommand{\defined}{\mathsf{defined}}
\newcommand{\call}{\mathsf{call}}
\newcommand{\fail}{\mathsf{fail}}
\newcommand{\eval}{\mathsf{eval}}
\newcommand{\clauses}{\mathsf{clauses}}
\newcommand{\midd}{\!\mid\!}
\newcommand{\sep}{\mathit{\;]\![\;}}
\newcommand{\success}{\mbox{\footnotesize \textsc{success}}}
\newcommand{\failsc}{\mbox{\footnotesize \textsc{fail}}}
\newcommand{\errorsc}{\mbox{\footnotesize \textsc{error}}}

\def\defemb#1#2{\expandafter\def\csname #1\endcsname
                              {\relax\ifmmode #2\else\hbox{$#2$}\fi}}
\defemb{cA}{{\cal A}}
\defemb{cB}{{\cal B}}
\defemb{cC}{{\cal C}}
\defemb{cD}{{\cal D}}
\defemb{cE}{{\cal E}}
\defemb{cF}{{\cal F}}
\defemb{cG}{{\cal G}}
\defemb{cH}{{\cal H}}
\defemb{cI}{{\cal I}}
\defemb{cJ}{{\cal J}}
\defemb{cL}{{\cal L}}
\defemb{cM}{{\cal M}}
\defemb{cN}{{\cal N}}
\defemb{cO}{{\cal O}}
\defemb{cP}{{\cal P}}
\defemb{cQ}{{\cal Q}}
\defemb{cR}{{\cal R}}
\defemb{cS}{{\cal S}}
\defemb{cT}{{\cal T}}
\defemb{cU}{{\cal U}}
\defemb{cV}{{\cal V}}
\defemb{cX}{{\cal X}}
\defemb{cZ}{{\cal Z}}

\newcommand{\var}{{\cV}ar}

\newcommand{\Var}{{\cal V}ar} 
\renewcommand{\emptyset}{\{\}} 
\renewcommand{\phi}{\varphi}

\newcommand{\ol}[1]{\overline{#1}}  

\def \tuple#1{\langle #1 \rangle}

\newcommand{\Dom}{\mathit{Dom}}
\newcommand{\Ran}{\mathit{Ran}}

\newcommand{\Hpos}{\mathcal{H}_{\mathit{pos}}}
\newcommand{\Hneg}{\mathcal{H}_{\mathit{neg}}}

\begin{document}

\title[Concolic Testing in Logic Programming]{Concolic Testing in Logic Programming%
  \thanks{This work has been partially supported by the EU (FEDER) and
    the Spanish \emph{Ministerio de Econom\'{\i}a y Competitividad}
    under grant TIN2013-44742-C4-1-R and by the \emph{Generalitat
      Valenciana} under grant PROMETEOII/2015/013. 
    Part of this research was done while the third author was visiting
    the University of Reunion; G.\ Vidal gratefully acknowledges their
    hospitality.}  }%

\author[F.~Mesnard and \'E.~Payet and G.~Vidal]
{FRED MESNARD, \'ETIENNE PAYET \\
  LIM - Universit\'e de la R\'eunion, France\\
  \email{\{fred,epayet\}@univ-reunion.fr} 
\and
GERM\'AN VIDAL \\
MiST, DSIC, Universitat Polit\`ecnica de Val\`encia\\
\email{gvidal@dsic.upv.es}
}


\maketitle

\label{firstpage}

\begin{abstract}
  Software testing is one of the most popular validation techniques in
  the software industry. Surprisingly, we can only find a few
  approaches to testing in the context of logic programming.
  In this paper, we introduce a systematic approach for dynamic
  testing that combines both concrete and symbolic execution. Our
  approach is fully automatic and guarantees full path coverage when
  it terminates.
  We prove some basic properties of our technique and illustrate its
  practical usefulness through a prototype implementation.\\
  
  \noindent
  \emph{To appear in Theory and Practice of Logic Programming (TPLP),
  Proc.\ of ICLP 2015.}\\
\end{abstract}


\section{Introduction} \label{intro}

Essentially, software validation aims at ensuring that the developed
software complies with the original requirements. One of the most
popular validation approaches is \emph{software testing}, a process
that involves producing a test suite and then executing the system
with these test cases. The main drawback of this approach is that
designing a test suite with a high code coverage |i.e., covering as
many execution paths as possible| is a complex and time-consuming
task. As an alternative, one can use a tool for the random generation
of test cases, but then we are often faced with a poor code coverage.
Some hybrid approaches exist where random generation is
driven by the user, as in QuickCheck \cite{CH00}, but then again the
process may become complex and time-consuming.

Another popular, fully automatic approach to test case generation is
based on \emph{symbolic execution} \cite{Kin76,Cla76}. Basically,
symbolic execution considers unknown (symbolic) values for the input
parameters and, then, explores all feasible execution paths in a
non-deterministic way. Symbolic states include now a \emph{path
  condition} that stores the current constraints on symbolic values,
i.e., the conditions that must hold to reach a particular execution
point. For each final state, a test case is produced by solving
the constraints in the associated path condition.

A drawback of the previous approach, though, is that the constraints
in the path condition may become very complex. When these constraints
are not solvable, the only sound way to proceed is to stop the
execution path, often giving rise to a poor coverage.  Recently, a new
variant called \emph{concolic execution} \cite{GKS05,SMA05} that
combines both \emph{conc}rete and symb\emph{olic} execution has been
proposed as a basis for both model checking and test case
generation. The main advantage is that, now, when the constraints in
the symbolic execution become too complex, one can still take some
values from the concrete execution to simplify them. This is sound and
often allows one to explore a larger execution space.  Some successful
tools that are based on concolic execution are, e.g., SAGE
\cite{GLM12} and Java Pathfinder \cite{PR10}.

In the context of the logic programming paradigm, one can find a
flurry of \emph{static}, complete techniques for software analysis and
verification. However, only a few dynamic techniques for program
validation have been proposed. Dynamic, typically incomplete,
techniques have proven very useful for software validation in other
paradigms. In general, these techniques are sound so that they avoid
\emph{false positives}. This contrasts with typical static
verification methods which may produce some false positives due to the
abstraction techniques introduced to ensure completeness. Therefore,
we expect concolic execution to complement 
existing analysis and verification techniques for logic programs.

In this paper, we introduce a new, fully automatic scheme for
\emph{concolic testing} in logic programming. As in other paradigms,
concolic testing may help the programmer to systematically find
program bugs and generate test cases with a good code coverage. As it
is common, our approach is always sound but usually incomplete.  In
the context of logic programming, we consider that ``full path
coverage'' involves calling each predicate in all possible
ways. Consider, e.g., the logic program $P=\{ p(a)., ~p(b).\}$. Here,
one could assume that the execution of the goals in $\{p(a),p(b)\}$ is
enough for achieving a full path coverage. However, in this paper we
consider that full path coverage requires, e.g., the set
$\{p(X),p(a),p(b),p(c)\}$ so that we have a goal that matches both
clauses, one that only matches the first clause, one that only matches
the second clause, and one that matches no clause. We call this notion
\emph{choice} coverage, and it is specific of logic programming. To
the best of our knowledge, such a notion of coverage has not been
considered before. Typically, only a form of \emph{statement} coverage
has been considered, where only the clauses used in the considered
executions are taken into account. For guaranteeing choice coverage, a
new type of unification problems must be solved:
we have to produce goals in which the selected atom $A$ matches
the heads of some clauses, say $H_1,\ldots,H_n$, but does not match
the heads of some other clauses, say $H'_1,\ldots,H'_m$.
We provide a constructive algorithm for solving such unifiability
problems.

A prototype implementation of the concolic testing scheme for pure
Prolog, called \textsf{contest}, is publicly available from
\texttt{http://kaz.dsic.upv.es/contest.html}. The results from an
experimental evaluation point out the usefulness of the
approach. Besides logic programming and Prolog, our technique might
also be useful for other programming languages since there exist
several transformational approaches that ``compile in'' programs to
Prolog, like, e.g., \cite{GZAP10}.

Omitted proofs as well as some extensions can be found in the
online appendix.

\section{Concrete Semantics} \label{concrete}

The semantics of a logic program is usually given in terms of the SLD
relation on goals \cite{Llo87}. In this section, we present instead a
\emph{local} semantics which is similar to that of
\citeN{SESGF11}. Basically, this semantics deals with states that
contain all the necessary information to perform the next step (in
contrast to the usual semantics, where the SLD tree built so far is
also needed, e.g., for dealing with the cut).
In contrast to \cite{SESGF11}, for simplicity, in this paper we only
consider definite logic programs.
However, the main difference w.r.t.\ \cite{SESGF11} comes from the
fact that our concrete semantics only considers the computation of the
first solution for the initial goal. This is the way most Prolog
applications are used and, thus, our semantics should consider this
behaviour in order to measure the coverage in a realistic way.

Before presenting the transition rules of the concrete semantics, let
us introduce some auxiliary notions and notations. We refer
  the reader to \cite{Apt97} for the standard definitions and
  notations for logic programs.
The semantics is defined by means of a transition system on
\emph{states} of the form
$\tuple{\cB^1_{\delta_1}\midd\ldots\midd\cB^n_{\delta_n}}$, where
$\cB^1_{\delta_1}\midd\ldots\midd\cB^n_{\delta_n}$ is a sequence of
goals labeled with substitutions (the answer computed so far, when
restricted to the variables of the initial goal).  We denote sequences
with $S, S',\ldots$, where $\epsilon$ denotes the empty sequence.  In
some cases, we label a goal $\cB$ both with a substitution and a
program clause, e.g., $\cB_\delta^{H\leftarrow\cB}$, which is used to
determine the next clause to be used for an SLD resolution step (see
rules \textsf{choice} and \textsf{unfold} in Fig.~\ref{fig:concrete2}).
Note that the clauses of the program are not included in the state but
considered a global parameter since they are static.
In the following, given an atom $A$
and a logic program $P$, $\clauses(A,P)$
returns the sequence of renamed apart program clauses $c_1,\ldots,c_n$ 
from $P$ whose head unifies with $A$.
A syntactic object $s_1$ is \emph{more general} than a syntactic object
$s_2$, denoted $s_1 \sleq s_2$, if there exists a substitution
$\theta$ such that $s_1\theta = s_2$.
$\var(o)$ denotes the set of variables of the syntactic object $o$. 
For a substitution $\theta$, $\var(\theta)$ is defined as $\Dom(\theta)\cup\Ran(\theta)$.

\begin{figure}[t]
  \rule{\linewidth}{1pt}
  \[
  \hspace{-2ex}\begin{array}{r@{~}l}
    \mathsf{(success)} & {\displaystyle 
      \frac{~} 
        {\tuple{\mathsf{true}_\delta\midd S} \to \tuple{\success_\delta}}
        } 
        \\[2ex]

     \mathsf{(failure)} & {\displaystyle 
      \frac{~} 
        {\tuple{(\fail,\cB)_\delta} \to \tuple{\failsc_\delta}}
        }
        \hspace{24ex}
    \mathsf{(backtrack)} ~ {\displaystyle 
      \frac{S\neq\epsilon} 
        {\tuple{(\fail,\cB)_\delta\midd S} \to \tuple{S}}
        }\\[4ex]

     \textsf{(choice)} &  {\displaystyle 
      \frac{
        \clauses(A,\cP) = (c_1,\ldots,c_n)\wedge n>0 
      } 
        {\tuple{(A,\cB)_\delta\midd S} \to
          \tuple{(A,\cB)_\delta^{c_1} \midd \ldots\midd
          (A,\cB)_\delta^{c_n}\midd  S}}
        } 
        \hfill\hspace{1ex}
    \mathsf{(choice\_fail)} ~ {\displaystyle 
      \frac{
        \clauses(A,\cP)=\emptyset} 
        {\tuple{(A,\cB)_\delta\midd S} \to \tuple{(\fail,\cB)_\delta\midd S}}
        }\\[3ex]

    \mathsf{(unfold)} & {\displaystyle 
      \frac{\mgu(A,H_1)=\sigma} 
        {\tuple{(A,\cB)_\delta^{H_1\leftarrow\cB_1}\midd S}
          \to \tuple{(\cB_1\sigma,\cB\sigma)_{\delta\sigma}\midd S}}
        } \\ 

    \end{array}
    \]
  \rule{\linewidth}{1pt}
  \caption{Concrete semantics} 
  \label{fig:concrete2}
\end{figure}

For simplicity, w.l.o.g., we only consider \emph{atomic} initial
goals. Therefore, given an atom $A$, an initial state has the form
$\tuple{A_\id}$, where $\id$ denotes the identity substitution.
The transition rules, shown in Figure~\ref{fig:concrete2}, proceed as
follows:
\begin{itemize}
\item In rules \textsf{success} and \textsf{failure}, we use fresh
  constants to denote a final state: $\tuple{\success_\delta}$ denotes
  that a sucessful derivation ended with computed answer substitution
  $\delta$, while $\tuple{\failsc_\delta}$ denotes a finitely failing
  derivation; recording $\delta$ for failing computations might be
  useful for debugging purposes.

\item Rule \textsf{backtrack} applies when the first goal in the
  sequence finitely fails, but there is at least one alternative
  choice. 

\item Rule \textsf{choice} represents the first stage of an SLD
  resolution step. If there is at least one clause whose head unifies
  with the leftmost atom, this rule introduces as many copies of a
  goal as clauses returned by function $\clauses$. If there is at
  least one matching clause, unfolding is then performed by rule
  \textsf{unfold}. Otherwise, if there is no matching clause, rule
  \textsf{choice\_fail} returns \textsf{fail} so that either rule
  \textsf{failure} or \textsf{backtrack} applies next.

\end{itemize}

\begin{example} \label{ex1}
  Consider the following logic program:
  \[
  \begin{array}{l@{~~~~~~}l@{~~~~~~}l}
    p(s(a)). &
    q(a).   &
    r(a).  \\
  
    p(s(X)) \leftarrow q(X). & q(b). & r(c). \\
    
    p(f(X)) \leftarrow r(X). & \\
    \end{array}
    \]
    Given the initial goal $p(f(X))$, we have the following successful
    computation (for clarity, we label each step with the applied
    rule):
    \[
    \begin{array}{llllll}
      \tuple{p(f(X))_\id} & \to^\mathsf{choice} &
      \tuple{p(f(X))_\id^{p(f(Y)) \leftarrow r(Y)}} & \to^\mathsf{unfold}
      & \tuple{r(X)_\id} \\
      & \to^\mathsf{choice} & \tuple{r(X)_\id^{r(a)}\midd r(X)_\id^{r(c)}}
      & \to^\mathsf{unfold}
      & \tuple{\mathsf{true}_{\{X/a\}}\midd r(X)_\id^{r(c)}} \\
      & \to^\mathsf{success} & \tuple{\success_{\{X/a\}}} 
    \end{array}
    \]
    Therefore, we have a successful computation for $p(f(X))$ with
    computed answer $\{X/a\}$. Observe that only the first answer is
    considered.
\end{example}
We do not formally prove the correctness of the concrete semantics,
but it is an easy consequence of the correctness of the semantics in
\cite{SESGF11}. Note that our rules can be seen as an instance for
pure Prolog without negation, where only the computation of the first
answer for the initial goal is considered.

\section{Concolic Execution Semantics} \label{concolic}

In this section, we introduce a concolic execution semantics for logic
programs that is a conservative extension of the concrete semantics of
the previous section. 
In this semantics, \emph{concolic states} have the form $\tuple{S\sep
  S'}$, where $S$ and $S'$ are sequences of (possibly labeled)
concrete and symbolic goals, respectively. In logic programming, the
notion of \emph{symbolic} execution is very natural: the structure of
both $S$ and $S'$ is the same, and the only difference is that some
atoms might be less instantiated in $S'$ than in $S$.

In the following, we let $\ol{o_n}$ denote the sequence of syntactic
objects $o_1,\ldots,o_n$.
Given an atom $A$, we let $\rt(A)= p/n$ if $A= p(\ol{t_n})$.
Now, given an atom $A$ with $\rt(A)=p/n$, an \emph{initial concolic
  state} has the form $\tuple{A_\id\sep p(\ol{X_n})_\id}$, where
$\ol{X_n}$ are different fresh variables. In the following, we assume
that every clause $c$ has a corresponding unique label, which we
denote by $\ell(c)$. By abuse of notation, we also denote by
$\ell(\ol{c_n})$ the set of labels $\{\ell(c_1),\ldots,\ell(c_n)\}$.

The semantics is given by the rules of the labeled transition
relation $\leadsto$ shown in Figure~\ref{fig:concolic}. Here, we
consider two kinds of labels for the transition relation:
\begin{itemize}
\item The empty label, $\diamond$, which is often implicit.
\item A label of the form $c(\ell(\ol{c_n}),\ell(\ol{d_k}))$, which
  represents a choice step. Here, $\ell(\ol{c_n})$ are the labels of
  the clauses matching the selected atom in the concrete goal, while
  $\ell(\ol{d_k})$ are the labels of the clauses matching the selected
  atom in the corresponding symbolic goal. Note that $\ell(\ol{c_n})
  \subseteq \ell(\ol{d_k})$ since the concrete goal is always an
  instance of the symbolic goal (see Theorem~\ref{th:invariant}
  below).
\end{itemize}
For each transition step $\cC_1 \leadsto_{c(\cL_1,\cL_2)} \cC_2$, the
first set of labels, $\cL_1$, is used to determine the execution
\emph{trace} of a concrete goal (see below). Traces are needed to keep
track of the execution paths already explored. The second set of
labels, $\cL_2$, is used to compute new goals that follow alternative
paths not yet explored, if any.

\begin{figure}[t]
  \rule{\linewidth}{1pt}\\[-1ex]
  \[
  \hspace{-2ex}\begin{array}{r@{~}l}
    \mathsf{(success)} & {\displaystyle 
      \frac{~~} 
        {\tuple{\mathsf{true}_\delta\midd S
            \sep\mathsf{true}_{\theta}\midd S'} 
          \leadsto_\diamond
          \tuple{\success_\delta \sep \success_\theta}}
        } \\[4ex]

    \mathsf{(failure)} & {\displaystyle 
      \frac{~} 
        {\tuple{(\mathsf{fail},\cB)_\delta
          \sep (\mathsf{fail},\cB')_\theta} \leadsto_\diamond 
        \tuple{\failsc_\delta \sep \failsc_\theta}}
        }\\[4ex]

    \mathsf{(backtrack)} & {\displaystyle 
      \frac{S\neq\epsilon} 
        {\tuple{(\mathsf{fail},\cB)_\delta\midd S \sep
          (\mathsf{fail},\cB')_\theta\midd S'} \leadsto_\diamond \tuple{S\sep S'}}
        }\\[4ex]

     \mathsf{(choice)} &  {\displaystyle 
      \frac{\clauses(A,\cP) = \ol{c_n}\wedge n>0 \wedge \clauses(A',\cP)=\ol{d_k}} 
        {\begin{array}{l@{~}l}
            \tuple{(A,\cB)_\delta\midd S \sep (A',\cB')_\theta\midd S'} 
            \leadsto_{\!c(\ell(\ol{c_n}),\ell(\ol{d_k}))} 
          \tuple{(A,\cB)_\delta^{c_1} \midd \ldots\midd
          (A,\cB)_\delta^{c_n}\midd  S\\[1ex]
        \hspace{47ex}\sep (A',\cB')_\theta^{c_1} \midd \ldots\midd
          (A',\cB')_\theta^{c_n}\midd  S'
        }\end{array}}
        }\\[8ex]

    \mathsf{(choice\_fail)} & {\displaystyle 
      \frac{ 
        \clauses(A,\cP)=\emptyset
      \wedge \clauses(A',\cP)=\ol{c_k}} 
        {\tuple{(A,\cB)_\delta\midd S \sep (A',\cB')_\theta\midd S'} 
          \leadsto_{c(\emptyset,\ell(\ol{c_k}))}
          \tuple{(\fail,\cB)_\delta\midd S\sep (\fail,\cB')_\theta\midd S'}}
        }\\[5ex]

    \mathsf{(unfold)} & {\displaystyle 
      \frac{\mgu(A,H_1)=\sigma\wedge\mgu(A',H_1)=\sigma'} 
        {\tuple{(A,\cB)_\delta^{H_1\leftarrow\cB_1}\midd S\sep
          (A',\cB')_\theta^{H_1\leftarrow\cB_1}\midd S'}
          \leadsto_\diamond \tuple{(\cB_1\sigma,\cB\sigma)_{\delta\sigma}\midd S
          \sep (\cB_1\sigma',\cB'\sigma')_{\theta\sigma'}\midd S'}}
        } \\ 

    \end{array}
    \]
  \rule{\linewidth}{1pt}
  \caption{Concolic execution semantics} \label{fig:concolic}
\end{figure}

In the concolic execution semantics, we perform both concrete and
symbolic execution steps in parallel. However, the symbolic execution
does not explore all possible execution paths but only mimics the
steps of the concrete execution; observe, e.g., rule \textsf{choice}
in Figure~\ref{fig:concolic}, where the clauses labeling the copies of
the symbolic goal are the same clauses $\ol{c_n}$ matching the
concrete goal, rather than the set of clauses $\ol{d_k}$ (a superset
of $\ol{c_n}$).

\begin{example} \label{ex2} 
  Consider again the logic program of Example~\ref{ex1}, now with
  clause labels:
  \[
  \begin{array}{l@{~~~~~~}l@{~~~~~~}l}
    (\ell_1)~ p(s(a)). &
    (\ell_4)~ q(a).   &
    (\ell_6)~ r(a).  \\
  
    (\ell_2)~ p(s(X)) \leftarrow q(X). & (\ell_5)~ q(b). & (\ell_7)~ r(c). \\
    
    (\ell_3)~ p(f(X)) \leftarrow r(X). & \\
    \end{array}
    \]
    Given the initial goal $p(f(X))$, we have the following concolic
    execution:
    \[
    \begin{array}{l@{~}l@{~}llll}
      \tuple{p(f(X))_\id\sep p(N)_\id} & \leadsto^\mathsf{choice}_{c(\cL_1,\cL'_1)} &
      \tuple{p(f(X))_\id^{p(f(Y)) \leftarrow r(Y)}\sep
        p(N)_\id^{p(f(Y)) \leftarrow r(Y)}} \\
      & \leadsto^\mathsf{unfold}_\diamond & \tuple{r(X)_\id\sep r(Y)_{\{N/f(Y)\}}} \\
      & \leadsto^\mathsf{choice}_{c(\cL_2,\cL'_2)} & \tuple{r(X)_\id^{r(a)}\midd
        r(X)_\id^{r(c)}\sep r(Y)_{\{N/f(Y)\}}^{r(a)}\midd
        r(Y)_{\{N/f(Y)\}}^{r(c)}} \\
      & \leadsto^\mathsf{unfold}_\diamond & \tuple{\mathsf{true}_{\{X/a\}}\midd
        r(X)_\id^{r(c)} \sep \mathsf{true}_{\{N/f(a)\}}\midd r(Y)_{\{N/f(Y)\}}^{r(c)}} \\
      & \leadsto^\mathsf{success}_\diamond & \tuple{\success_{\{X/a\}}\sep \success_{\{N/f(a)\}}} 
    \end{array}
    \]
    where $\cL_1 = \{\ell_3\}$, $\cL'_1=\{\ell_1,\ell_2,\ell_3\}$, and
    $\cL_2 = \cL'_2 = \{\ell_6,\ell_7\}$.
\end{example}
In this paper, we only consider finite concolic executions for initial
goals.  This is a reasonable assumption since one can expect concrete
goals to compute the first answer finitely (unless the program is
erroneous).
We associate a \emph{trace} to each concolic execution as follows:

\begin{definition}[trace]
  Let $P$ be a program and $\cC_0$ an initial concolic state. Let $E =
  (\cC_0 \leadsto_{l_1} \ldots \leadsto_{l_m} \cC_m)$, $m>0$, be a
  concolic execution for $\cC_0$ in $P$. Let
  $c(\cL_1,\cL'_1),\ldots,c(\cL_k,\cL'_k)$, $k\sleq m$, be the
  sequence of labels in $l_1,\ldots,l_m$ which are different from
  $\diamond$. Then, the trace associated to the concolic execution $E$
  is $\trace(E) = \cL_1,\ldots,\cL_k$.
\end{definition}
Roughly speaking, a trace is just a sequence with the sets of labels
of the matching clauses in each choice step.
For instance, the trace associated to the concolic execution of
Example~\ref{ex2} is $(\{\ell_3\},\{\ell_6,\ell_7\})$, i.e., we have
two unfolding steps with matching clauses $\{\ell_3\}$ and
$\{\ell_6,\ell_7\}$, respectively.
Note that traces ending with $\{\:\}$ represent failing derivations.

The following result states an essential invariant for concolic execution:

\begin{theorem} \label{th:invariant} Let $P$ be a program and $\cC_0 =
  \tuple{p(\ol{t_n})_\id \sep p(\ol{X_n})_\id }$ be an initial
  concolic state. Let $\cC_0 \leadsto \ldots \leadsto \cC_m$, $m\sgeq 0$,
  be a finite (possibly incomplete) concolic execution for $\cC_0$ in
  $P$. Then, for all concolic states $C_i= \tuple{\cB^c_\delta\midd
    S\sep \cD^{c'}_\theta\midd S'}$, $i=0,\ldots,m$, the following
  invariant holds: $|S|=|S'|$, $\cD\sleq \cB$, $c=c'$ (if any), and
  $p(\ol{X_n})\theta \sleq p(\ol{t_n})\delta$.
\end{theorem}

\section{Concolic Testing} \label{testing}

In this section, we introduce a concolic testing procedure for logic
programs based on the concolic execution semantics of the previous
section.

\subsection{The Procedure}

As we have seen in Section~\ref{concolic}, the concolic execution
steps labeled with $c(\cL_1,\cL_2)$ give us a hint of (potential)
alternative execution paths. Consider, for instance, the concolic
execution of Example~\ref{ex2}. The first step is labeled with
$c(\{\ell_3\},\{\ell_1,\ell_2,\ell_3\})$. This means that the selected
atom in the concrete goal only matched clause $\ell_3$, while the
selected atom in the symbolic goal matched clauses $\ell_1$, $\ell_2$
and $\ell_3$. In principle, there are as many alternative execution
paths as elements in
$\cP(\{\ell_1,\ell_2,\ell_3\})\:\backslash\:\{\ell_3\}$; e.g.,
$\{\:\}$ denotes an execution path where the selected atom matches no
clause, $\{\ell_1\}$ another path in which the selected atom
\emph{only} matches clause $\ell_1$, $\{\ell_1,\ell_2,\ell_3\}$
another path where the selected atom matches all three clauses
$\ell_1$, $\ell_2$ and $\ell_3$, and so forth.

When aiming at full choice coverage, we need to
solve both unification and disunification problems. Consider, e.g.,
that $A$ is the selected atom in a goal, and that we want it to unify
with the head of clause $\ell_1$ but not with the heads of clauses
$\ell_2$ and $\ell_3$. For this purpose, we introduce the following
auxiliary function $\alt$, which also includes some groundness
requirements (see below). In the following, we let $\approx$ denote
the unifiability relation, i.e., given atoms $A,B$, $A\approx B$ holds
if $\mgu(A,B)\neq\fail$; correspondingly, $\neg(A\approx B)$ holds if
$\mgu(A,B)=\fail$.

\begin{definition}[$\alt$] \label{def:alt}
  Let $A$ be an atom and $\cL,\cL'$ be sets of clause labels. Let
  $\cV$ be a set of variables.  The function $\alt(A,\cL,\cL',\cV)$
  returns a substitution $\theta$ such that the following holds:
  \[
  A\theta\approx H_1\wedge\ldots\wedge A\theta\approx H_n\wedge
  \neg(A\theta\approx H_{n+1})\wedge\ldots\wedge \neg(A\theta\approx
  H_m)
  \wedge \cV\theta~\mbox{are ground}
  \]
  where $H_1,\ldots,H_n$ are the heads of the (renamed apart) clauses
  labeled by $\cL$ and $H_{n+1},\ldots,H_m$ are the heads of the
  (renamed apart) clauses labeled by $\cL'\backslash\cL$,
  respectively. 
  If such a substitution does not exist, then function $\alt$ returns
  $\fail$.
\end{definition}
When the considered signature is finite,\!\footnote{Full Prolog and
  infinite signatures like integers or real numbers are left as future
  work.}  the following semi-algorithm is trivially sound and complete
for solving the above unifiability problem: first, bind $A$ with terms
of depth $0$.\footnote{The depth $\depth(t)$ of a term $t$ is defined
  as usual: $\depth(t) = 0$ if $t$ is a variable or a constant symbol,
  and $\depth(f(t_1,\ldots,t_n)) =
  1+\mathsf{max}(\depth(t_1),\ldots,\depth(t_n))$, otherwise.} If the
condition above does not hold, then we try with terms of depth $1$,
and check it again. We keep increasing the considered term depth until
a solution is found. If a solution exists, this naive semi-algorithm will
find it (otherwise, it may run forever). 
In practice, however, it may be very inefficient.

Observe that, in general, there might be several most general solutions to
the above problem. Consider, e.g., $A=p(X,Y)$,
$\Hpos=\{p(Z,Z),p(a,b)\}$ and $\Hneg=\{p(c,c)\}$. Then, both $p(a,U)$
and $p(U,b)$ are most general solutions.  In principle, any of them is
equally good in our context.
We postpone to the next section the introduction of a constructive
algorithm for function $\alt$.  Here, we present an algorithm to
systematically produce concrete initial goals so that all
\emph{feasible} choices in the execution paths are covered (unless the
process runs forever).
First, we introduce the following auxiliary definitions:

\begin{definition}[$\conc$, $\symb$]
\label{defn:conc:symb}
  Let $\cC =
  \tuple{\cB^1_{\delta_1}\midd\ldots\midd\cB^n_{\delta_n}\sep
    \cD^1_{\theta_1}\midd\ldots\midd\cD^n_{\theta_n}}$ be a concolic
  state. Then, we let $\conc(\cC) = \cB^1_{\delta_1}$ denote the
  first concrete goal and $\symb(\cC) = \cD^1_{\theta_1}$ the
  first symbolic goal.
\end{definition}

\begin{definition}[$\alttrace$] \label{def:alttrace}
  Let $P$ be a program, $\cC_0$ an initial concolic state, and $E =
  (\cC_0 \leadsto_{l_1} \ldots \leadsto_{l_n} \cC_n
  \leadsto_{c(\cL,\cL')} \cC_{n+1})$ be a (possibly incomplete)
  concolic execution for $\cC_0$ in $P$.  Then, the function
  $\alttrace$ denotes the following set of (potentially) alternative
  traces:
  \[
  \begin{array}{lll}
  \alttrace(E) = \{ \cL_1,\ldots,\cL_k,\cL'' \mid & \trace(\cC_0
  \leadsto_{l_1} \ldots \leadsto_{l_n} \cC_n) =
  \cL_1,\ldots,\cL_k & \\
  & \mbox{and}~\cL''\in(\cP(\cL')\:\backslash\:\cL) & \}
  \end{array}
  \]
\end{definition}
For instance, given the following (partial) concolic execution $E$
from Example~\ref{ex2}:
\[\hspace{-2ex}
\begin{array}{lll}
\tuple{p(f(X))_\id\sep p(N)_\id} & \leadsto^\mathsf{choice}_{c(\cL_1,\cL'_1)} 
      & \tuple{p(f(X))_\id^{p(f(Y)) \leftarrow r(Y)}\sep
        p(N)_\id^{p(f(Y)) \leftarrow r(Y)}} \\
      & \leadsto^\mathsf{unfold}_\diamond & \tuple{r(X)_\id\sep
        r(Y)_{\{N/f(Y)\}}} \\ 
      & \leadsto^\mathsf{choice}_{c(\cL_2,\cL'_2)} & \tuple{r(X)_\id^{r(a)}\midd
        r(X)_\id^{r(c)}\sep r(Y)_{\{N/f(Y)\}}^{r(a)}\midd
        r(Y)_{\{N/f(Y)\}}^{r(c)}}
    \end{array}
\]
where $\cL_1=\{\ell_3\}$, $\cL'_1=\{\ell_1,\ell_2,\ell_3\}$,
$\cL_2=\cL'_2 = \{\ell_6,\ell_7\}$, we have $\trace(E)=\cL_1,\cL_2$,
$\cP(\cL'_2)\:\backslash\:\cL_2=\{ \{\:\},\{\ell_6\},\{\ell_7\} \}$, and
$\alttrace(E) =\{(\cL_1,\{\:\}),(\cL_1,\{\ell_6\}),(\cL_1,\{\ell_7\})\}$.

Now, we introduce our concolic testing procedure. It takes as input a
program and a random |e.g., provided by the user| initial atomic goal
rooted by the distinguished predicate $main/n$. In the following, we assume
that each concrete initial goal $main(\ol{t_n})$ is existentially
terminating w.r.t.\ Prolog's leftmost computation rule, i.e., either
computes the \emph{first} answer in a finite number of steps or
finitely fails \cite{Vasak86a}.
For this purpose, we assume that $main/n$ has some associated
\emph{input} arguments, determined by a function $\inp$, so that an
initial goal $main(\ol{t_n})$ existentially terminates if the terms
$\inp(main(\ol{t_n}))$ are ground. One could also consider that there
are several combinations of input arguments that guarantee existential
termination |this is similar to the modes of a predicate| but we only
consider one set of input arguments for simplicity (extending the
concolic testing algorithm would be straightforward).
As mentioned before, assuming that concrete initial goals are
existentially terminating is a reasonable assumption in practice.

\begin{definition}[concolic testing] \label{def:concolic}
  \begin{description}
  \item[\textbf{Input:}] a logic program $P$ and an atom
    $main(\ol{t_n})$ with $\inp(main(\ol{t_n}))$ ground.

  \item[\textbf{Output:}] a set $TC$ of test cases.
    
  \end{description}

  \begin{enumerate}
  \item Let $\mathit{Pending} :=\{main(\ol{t_n})\}$, $TC:=\{\}$,
    $\traces:=\{\}$.

  \item While $|Pending|\neq 0$ do
    \begin{enumerate}
    \item Take $A\in \mathit{Pending}$, $\mathit{Pending} :=
      \mathit{Pending} \backslash \{A\}$, $TC := TC\cup \{A\}$.

    \item 
      Let $\cC_0 = \tuple{A_\id\sep main(\ol{X_n})_\id}$ and 
      compute a successful or finitely failing derivation
      $
      E = (\cC_0 \leadsto_{l_1} \ldots \leadsto_{l_m} \cC_m)
      $. 
      
    \item Let $\traces :=
      \traces\cup\{\trace(E)\}$. 

    \item We update $\mathit{Pending}$ 
      as follows:
      \begin{itemize}
      \item for each prefix $ \cC_0 \leadsto_{l_1} \ldots
        \leadsto_{l_j} \cC_j \leadsto_{c(\cL,\cL')} \cC_{j+1} $ of
        $E$ and
      \item for each (possibly partial) trace $\ol{\cL_{k}},\cL_{k+1}$
        $\in\alttrace(\cC_0 \leadsto_{l_1} \ldots \leadsto_{l_j} \cC_j
        \leadsto_{c(\cL,\cL')} \cC_{j+1})$ which is not the prefix of
        any trace in $\traces$, 
      \item add 
        $main(\ol{X_n})\theta\theta'$ to $\mathit{Pending}$ 
        if $\alt(A_1,\cL_{k+1},\cL',G)=\theta'\neq\fail$, where
        $G=\var(\inp(main(\ol{X_n})\theta))$ and
        $\symb(\cC_j)= (A_1,\cB)_\theta$.\!\footnote{I.e., $A_1$ is
          the first atom of the symbolic goal $\symb(\cC_j)$ of the
          concolic state $\cC_j$, see Definition
          \ref{defn:conc:symb}. }
      \end{itemize}
    \end{enumerate}
  \item Return the set $TC$ of test cases
  \end{enumerate}
\end{definition}
The soundness of concolic testing is immediate, since 
each atom from $TC$ is indeed a test case 
of the form $main(\ol{s_n})$ with $\inp(main(\ol{s_n}))$ ground. 
Completeness and termination are more subtle properties though.

In principle, one could argue that the concolic testing algorithm is
a complete semi-algorithm in the sense that, if it terminates, the
generated test cases cover all feasible paths. Our assumptions
trivially guarantee that every considered concrete execution is finite
(i.e., step (2b) in the loop of the concolic testing
algorithm). Unfortunately, the algorithm will often run forever by
producing infinitely many test cases. Consider, e.g., the following
simple program:
\[
(\ell_1) ~nat(0).\hspace{10ex} (\ell_2)~nat(s(X)) \leftarrow nat(X).
\]
Even if every goal $nat(t)$ with $t$ ground is terminating, our
algorithm will still produce infinitely many test cases, e.g.,
$nat(0)$, $nat(s(0))$, $nat(s(s(0)))$, \ldots, since each goal will
explore a different path (i.e., will produce a different execution
trace: $(\{\ell_1\})$, $(\{\ell_2\},\{\ell_1\})$,
$(\{\ell_2\},\{\ell_2\},\{\ell_1\})$, etc). In practice, though, the
quality of the generated test cases should be experimentally evaluated
using a coverage tool.

Therefore, in general, we will sacrifice completeness in order to
guarantee the termination of concolic testing. For this purpose, one
can use a time limit, a bound for the length of concolic executions,
or a maximum term depth for the arguments of the generated test
cases. In this paper, we consider the last approach.
Then, one can replace the use of a particular function $\alt$ in step
(2d) of Definition~\ref{def:concolic} by a function $\alt_k$ with
$\alt_k(A,\cL,\cL',G)=\alt(A,\cL,\cL',G)=\theta$ if $depth(t)\sleq k$
for all $X/t\in\theta$, and $\alt_k(A,\cL,\cL',G) =\fail$ otherwise.
This is the solution we implemented in the concolic testing tool
described in Section~\ref{tool}.

For instance, by requiring a maximum term depth of $1$, the
generated test cases for the program $nat$ above would be $nat(0)$,
$nat(1)$, $nat(s(0))$ and $nat(s(1))$, where $1$ is a fresh constant
symbol, with associated traces $(\{\ell_1\})$, $(\{\:\})$,
$(\{\ell_2\},\{\ell_1\})$, and $(\{\ell_2\},\{\:\})$, respectively.

Termination of the algorithm in Definition~\ref{def:concolic} is then
guaranteed since only a finite number of new atoms can be added in
step (2d) |up to variable renaming| and, moreover, only those (possibly
partial) traces which are not a prefix of any trace already in the set
$\mathit{Traces}$ are considered. Observe that these facts suffice to
ensure termination of the algorithm since one cannot have infinitely
many traces with a finite number of atoms.

\subsection{Solving Unifiability Problems} \label{sec:unif}

In this section, we present a constructive algorithm for function $\alt$.
Let us first reformulate our unification problem in slightly more
general terms than in Definition~\ref{def:alt}. Let $A$ be an atom and
$\Hpos$, $\Hneg$ be two sets of atoms the elements of which are
variable disjoint with $A$ and unify with $A$, and a set of variables
$G$. The problem consists in finding a substitution $\sigma$ such that
\[
\forall H^+\in\Hpos.~
A\sigma\approx H^+ \wedge 
\forall H^-\in\Hneg.~
\neg (A\sigma\approx H^-),~\mbox{and}~ G\sigma ~\mbox{is ground}\hspace{2ex}(*)
\]
We introduce a stepwise method that, roughly speaking, proceeds as
follows:
\begin{itemize}
\item First, we produce some ``maximal'' substitutions $\theta$
  (called \emph{maximal unifying substitution} below) for $A$ such
  that $A\theta$ still unifies with the atoms in $\Hpos$. Here, we use
  a special set $\cU$ of fresh variables with
  $\var(\{A\}\cup\Hpos\cup\Hneg)\cap\cU=\emptyset$. The elements of
  $\cU$ are denoted by $U$, $U'$, $U_1$\dots{} Then, in $\theta$, the
  variables from $\cU$ (if any) denote positions where further binding
  will prevent $A\theta$ from unifying with some atom in $\Hpos$. In
  contrast, $A\theta\sigma'$ still unifies with all the atoms in
  $\Hpos$ as long as $\sigma'$ does not bind any variable from $\cU$.
  Roughly speaking, we apply some (minimal) generalizations to the
  atoms in $\Hpos$ so that they unify, and then return their most
  general unifier.

  For this stage, we use well known techniques like \emph{variable
    elimination} \cite{MM82}
  and \emph{generalization} (from the algorithm for most specific
  generalization \cite{Plo70}); see Definition~\ref{alg1} below.

\item In a second stage, we look for another substitution $\eta$
  such that $\theta\eta$ is a solution for $(*)$. Here, we
  basically follow a generate and test algorithm (as in the naive
  algorithm above), but it is now much more restricted thanks to
  $\theta$.
\end{itemize}

\subsubsection{The Positive Atoms}
\label{section-algo-all-pos}

Here, we will use the variables from the special set $\cU$ to
replace \emph{disagreement pairs}
(see~\cite{Apt97} p.~27).
Roughly speaking, given terms $s$ and $t$, a subterm $s'$ of $s$ and a
subterm $t'$ of $t$ form a disagreement pair if the root symbols of
$s'$ and $t'$ are different, but the symbols from $s'$ up to the root of
$s$ and from $t'$ up to the root of $t$ are the same. For instance,
$X,g(a)$ and $b,h(Y)$ are disagreement pairs of the terms $f(X,g(b))$
and $f(g(a),g(h(Y)))$.
A disagreement pair $t,t'$ is called \emph{simple}
if one of the terms is a variable that does not occur in the other
term and no variable of $\cU$ occurs in $t,t'$.  We say that
the substitution $\{X/s\}$ is determined by $t,t'$ if
$\{X,s\}=\{t,t'\}$.

Basically, given an atom $A$ and a set of atoms $\Hpos$, the following
algorithm nondeterministically computes a substitution $\theta$ such
that $A\theta\sigma'$ still unifies with all the atoms in $\Hpos$ as
long as $\sigma'$ does not bind any variable from $\cU$.

\begin{definition}[maximal unifying substitution] \label{alg1}
\begin{description}
\item[\textbf{Input:}] 
  an atom $A$ and a set of atoms $\Hpos$ such that
  $\var(\{A\}\cup\Hpos)\cap\cU=\emptyset$ and $A\approx B$ for all $B\in\Hpos$.
\item[\textbf{Output:}] a substitution $\theta$.
\end{description}

\begin{enumerate}
\item \label{algo-msa-init}
  Let $\cB:=\{A\}\cup\Hpos$.
\item \label{algo-msa-while-simple}
  While simple disagreement pairs occur in $\cB$ do
  \begin{enumerate}
  \item nondeterministically choose a simple disagreement pair $X,t$
    (resp.\ $t,X$) in $\cB$ such that there is no other simple
    disagreement pair of the form $X,t'$ (or $t',X$) with $t<t'$
    (i.e., a strict instance);
  \item \label{algo-msa-simple-pair}
    set $\cB$ to $\cB\eta$ where $\eta = \{X/t\}$.
  \end{enumerate}
\item \label{algo-msa-while-not-simple}
  While $|\cB|\neq 1$ do
  \begin{enumerate}
  \item nondeterministically choose a 
    disagreement pair $t,t'$ in $\cB$;
  \item \label{algo-msa-not-simple-pair} replace all disagrement
    pairs $t,t'$ in $\cB$ by a fresh variable of $\cU$.
  \end{enumerate}
\item \label{algo-msa-return} 
  Return $\theta$, where $\cB=\{B\}$, $A\theta = B$, and $\Dom(\theta)\subseteq\Var(A)$.
\end{enumerate}
\end{definition}
We note that the algorithm assumes that the input atom $A$ is always more
general than the final atom $B$ so that the last step is well defined. 
An invariant proving that this is indeed the case can be found in the
online appendix (Appendix B).

Observe that the step (2a) is nondeterministic since there may exist
several disagreement pairs $X,t$ (or $t,X$) for the same variable $X$.
Consider the atom $A = p(X,Y)$ and the set
$\Hpos=\{p(a,b),p(Z,Z)\}$. Then, both $\{X/a,Y/U\}$ and $\{X/U,Y/b\}$
are maximal unifying substitutions, as the following example
illustrates:

\begin{example}\label{ex-1-g}
  Let $A=p(X,Y)$ and $\Hpos=\{p(a,b),p(Z,Z)\}$, with
  $\cB:=\{p(X,Y),p(a,b),p(Z,Z)\}$. 
  The algorithm then considers
  the simple disagreement pairs in $\cB$. From $X,a$, we get
  $\eta_1:=\{X/a\}$ and the
  action~(\ref{algo-msa-simple-pair}) sets $\cB$ to
    $\cB\eta_1=\{p(a,Y),p(a,b),p(Z,Z)\}$.
   The substitution $\eta_2:=\{Y/b\}$ is determined by $Y,b$ and the
     action~(\ref{algo-msa-simple-pair}) sets $\cB$ to
     $\cB\eta_2=\{p(a,b),p(Z,Z)\}$.
   Now, we have two non-deterministic possibilities:
   \begin{itemize}
   \item If we consider the disagreement pair $a,Z$, we have a
     substitution $\eta_3:=\{Z/a\}$ and
     Action~(\ref{algo-msa-simple-pair}) then sets $\cB$ to
     $\cB\eta_3=\{p(a,b),p(a,a)\}$. Now, no simple disagreement pair
     occurs in $\cB$, hence the algorithm jumps to the loop at
     line~\ref{algo-msa-while-not-simple}.
     Action~(\ref{algo-msa-not-simple-pair}) replaces the disagreement
     pair $b,a$ with a fresh variable $U\in \cU$, hence $\cB$ is set
     to $\{p(a,U)\}$.  As $|\cB|=1$ the loop at
     line~\ref{algo-msa-while-not-simple} stops and the algorithm
     returns the substitution $\{X/a,Y/U\}$.
   \item If we consider the disagreement pair $b,Z$ instead, we have  a
     substitution $\eta'_3:=\{Z/b\}$, and
     Action~(\ref{algo-msa-simple-pair}) sets $\cB$ to
     $\cB\eta'_3=\{p(a,b),p(b,b)\}$. Now, by proceeding as in the
     previous case, the algorithm returns 
     $\{X/U,Y/b\}$.
   \end{itemize}
\end{example}

\subsubsection{The Negative Atoms}
\label{section-algo-all-pos-neg}

Now we deal with the negative atoms 
by means of the following
algorithm which is the basis of our implementation of function
$\alt$:

\begin{definition}[PosNeg] \label{alg2}
\begin{description}
\item[\textbf{Input:}] an atom $A$ and two sets of atoms
  $\Hpos$, $\Hneg$, the elements of which are variable
  disjoint with $A$ and unify with $A$, and
  a set of variables $G$. 
\item[\textbf{Output:}] $\fail$ or a substitution $\theta\eta$
  (restricted to the variables of $A$).
\end{description}

\begin{enumerate}\itemsep2pt
\item Let $\theta$ be the substitution returned by the algorithm
  of Definition~\ref{alg1}
  with input $A$ and $\Hpos$.
\item Let $\eta$  be an idempotent substitution such that
  $G\theta\eta$ is ground.
\item Check that 
  $\Dom(\eta)\subseteq\Var(A\theta)$ and
  $\Var(\eta)\cap \cU=\emptyset$,
  otherwise return $\fail$.
\item Check that for each $H^-\in\Hneg$, 
  $\neg (A\theta\eta\approx H^-)$,
  otherwise return $\fail$.
\item Return $\theta\eta$
  (restricted to the variables of $A$). 
\end{enumerate}
\end{definition}
The correctness of this algorithm is stated as follows:

\begin{theorem}\label{theorem:correction-algo-posneg}
  Let $A$ be an atom and $\Hpos,\Hneg$ be two sets of atoms such that
  $\var(\{A\}\cup\Hpos\cup\Hneg)\cap\cU=\emptyset$ and $A\approx B$
  for all $B\in\Hpos\cup\Hneg$, 
  and a set of variables $G$.
  The algorithm in Definition~\ref{alg2} always terminates and, if it
  returns a substitution $\sigma$, then $\bigwedge_{H\in\Hpos}
  A\sigma\approx H\wedge \bigwedge_{H'\in\Hneg} \neg(A\sigma\approx
  H')$ holds and $G\sigma$ is ground.
\end{theorem}

\begin{example}
  Let $A:=p(X)$, $\Hpos:=\{p(s(Y))\}$, $\Hneg:=\{p(s(0))\}$,  
  and $G:=\{X\}$.
  The algorithm of Definition~\ref{alg1} 
  returns $\theta = \{X/s(Y)\}$. 
  We take $\eta=\{Y/s(0)\}$, it is idempotent and  
  $G\theta\eta$ is ground.
  $\Dom(\eta)\subseteq \Var(A\theta)$ 
  and $\Var(\eta)=\{Y\}$ does not intersect with $\cU$.
  Finally, $A\theta\eta=p(s(s(0)))$ does not unify with $p(s(0))$.
  The algorithm thus returns $\theta\eta=\{X/s(s(0)),Y/s(0)\}$
  restricted to the variables of $A$, i.e., 
  $\{X/s(s(0))\}$.
\end{example}

\begin{example}
  Let $A:=p(X)$, $\Hpos:=\{p(a),p(b)\}$, 
  $\Hneg:=\{p(f(Z))\}$, and $G:=\emptyset$. 
  The algorithm of Definition~\ref{alg1}
  applied to $A$ and $\Hpos$ 
  returns $\theta=\{X/U\}$. 
  However, we cannot find $\eta$ such that $A\theta\eta$ does not
  unify with $p(f(Z))$ without binding $U$. 
  The algorithm thus returns $\fail$.  
\end{example}
Theorem~\ref{theorem:correction-algo-posneg} states the soundness of our
procedure for computing function $\alt$.
As for completeness, we claim that binding an atom $A$ with all
possible maximal unifying substitutions for $A$ and $\Hpos$ does not
affect to the existence of a solution to the unification problem (*)
above
(see the online appendix (Appendix B) for more details).

\subsection{A Tool for Concolic Testing} \label{tool}

In this section, we present a prototype implementation of the
concolic testing scheme. The tool, called \textsf{contest}, is
publicly available from the following URL
\begin{quote}
  \texttt{http://kaz.dsic.upv.es/contest.html}
\end{quote}
It consists of approx.\ 1000 lines of Prolog code and implements the
concolic testing algorithm of Definition~\ref{def:concolic} with
function $\alt$ as described in Section~\ref{sec:unif} and a maximum
term depth that can be fixed by the user in order to guarantee the
termination of the process. 
Moreover, we also introduced a bound for the number of alternatives
when computing function $\alttrace$ (see
Definition~\ref{def:alttrace}). Roughly speaking, when the number of
alternatives is too high, we give up aiming at full \emph{choice}
coverage and return sets with only one clause label (which suffice for
clause coverage).

\begin{table}[t]
  \caption{Clause coverage analysis results (SICStus Prolog)}   \label{table:coverage}
  \centering
  \footnotesize
  $
  \begin{array}{l@{~}r|l@{~}r|l@{~}r|l@{~}r|l@{~}r}
    \hline\hline
   paper & 100\% &
   paper2 & 100\% &
   nat & 100\% &
   advisor & 100\% &
   applast & 100\% \\ 
   depth & 88\% &
   regexp & 86\% &
   relative & 100\% &
   rotateprune & 100\% &
   transpose & 100\% \\ 
   mult & 100\% &
   hanoi & 100\% &
   automaton & 100\% &
   qsort & 95\% &
   inclist & 100\% \\ 
   doubleflip & 100\% &
   recacctype & 100\% &
   ackermann & 100\% &
   fibonacci & 100\% &
   preorder & 100\% \\\hline\hline
  \end{array}
  $
\end{table}

Table~\ref{table:coverage} shows a summary of the coverage achieved by
the test cases automatically generated using \textsf{contest}. The
complete benchmarks --including the source code, initial goal, input
arguments and maximum term depth-- 
can be found in the above URL.
We used the coverage analysis tool of SICStus Prolog 4.3.1, which
basically measures the number of times each clause is used. The
results are very satisfactory, achieving a full coverage in most of
the examples.

The current version is a proof-of-concept implementation and only
deals with pure Prolog without negation. We plan to extend it to cope
with full Prolog. The concrete semantics can be extended following
\cite{SESGF11}, and concolic execution is in general a natural
extension of the semantics in Figure~\ref{fig:concolic}. 
For relational built-in's or
equalities, we should label the execution step with an associated
constraint, which can then be used to produce alternative execution
paths by solving its negation. In this context, our tool will be
useful not only for test case generation, but also to detect program
errors during concolic testing (e.g., negated atoms which are not
instantiated enough, incorrect calls to arithmetic built-in's, etc).
See the online appendix (Appendix A) for more details on extending
concolic execution to full Prolog.

\section{Related Work and Concluding Remarks} \label{relwork}

\citeN{DBLP:conf/iclp/MeraLH09} 
present a framework unifying unit testing and 
run-time verification for the 
Ciao system~\cite{DBLP:journals/tplp/HermenegildoBCLMMP12}.
The $\mathrm{ECL}^i\mathrm{PS}^e$ constraint programming 
system~\cite{DBLP:journals/tplp/SchimpfS12}
and SICStus Prolog~\cite{DBLP:journals/tplp/CarlssonM12}
both provide  tools which run a given goal and compute
how often program points in the code were executed.
SWI-Prolog~\cite{wielemaker:2011:tplp} offers a unit testing tool
associated to an optional interactive generation of test cases.
It also includes an experimental coverage analysis which runs a given goal
and computes the percentage of the used clauses and failing clauses.
\citeN{DBLP:conf/issta/BelliJ93} and \citeN{DBLP:conf/lopstr/DegraveSV08}
consider automatic generation of test inputs for  strongly typed and moded logic
programming languages like the Mercury programming 
language~\cite{Somogyi96a}, whereas we only require moding the top-level predicate of the program.

One of the closest approaches to our work is the test case generation
technique by \cite{AAGR14}. The main difference, though, is that their
technique is based solely on traditional symbolic execution. As
mentioned before, concolic testing may scale better since one can deal
with more complex constraints by using data from the concrete
component of the concolic state.  
Another difference is that we aim at full path coverage (i.e., choice
coverage), and not only a form of statement coverage.

Another close approach is \cite{Vid15}, where a concolic execution
semantics for logic programs is presented. However, this approach only
considers a simpler statement coverage and, thus, it can be seen as a
particular instance of the technique in the present paper. Another
significant difference is that, in \cite{Vid15}, concolic execution
proceeds in a stepwise manner: first, concrete execution produces an
execution \emph{trace}, which is then used to drive concolic
execution. Although this scheme is conceptually simpler, it may give
rise to poorer results in practice since one cannot use concrete
values in symbolic executions, one of the main advantages of concolic
execution over traditional symbolic execution.  Moreover,
\citeN{Vid15} presents no formal results nor an implementation of the
concolic execution technique.

Summarizing the paper, we have introduced a novel scheme for concolic testing in logic
programming. It offers a sound and fully automatic technique for
test case generation with a good code coverage. 
We have stated 
the particular type of unification problems that should be solved to produce new test cases. 
We have proposed a correct algorithm for such unification problems.
Furthermore, we have developed a publicly available
proof-of-concept implementation of the concolic testing scheme,
\textsf{contest}, that shows the usefulness of our approach.
To the best of our knowledge, this is the first fully automatic
testing tool for Prolog that aims at full path coverage (here called choice coverage).

As future work, we plan to extend the scheme to full Prolog (see the
remarks in Section~\ref{tool}).
Another interesting subject for further research is the definition of
appropriate heuristics to drive concolic testing w.r.t.\ a given
coverage criterion. This might have a significant impact on the
quality of the test cases when the process is incomplete.
Finally, from the experimental evaluation, we observed that the
results could be improved by introducing type information, so that the
generated values are restricted to the right type. Hence, improving concolic
testing with type annotations is also a promising line of future work.




\clearpage

\begin{appendix}

\pagestyle{plain}

\mbox{}\\[1ex]

  \begin{center}
    {\large\textnormal{Online appendix for the paper}}   \\
    \LARGE\em Concolic Testing in Logic Programming\\
    {\large\textnormal{published in Theory and Practice of Logic Programming}}
  \end{center}

\vspace{1ex}

  \begin{center}
    FRED MESNARD, \'ETIENNE PAYET \\
\em LIM - Universit\'e de la R\'eunion, France\\
\email{\{fred,epayet\}@univ-reunion.fr} \\[1ex]
GERM\'AN VIDAL \\
\em MiST, DSIC, Universitat Polit\`ecnica de Val\`encia\\
\email{gvidal@dsic.upv.es}
  \end{center}

\vspace{.5ex}

  \begin{center}
    \emph{ submitted 29 April 2015; revised 3rd July 2015; accepted 14
      July 2015 } 
  \end{center}

\vspace{1.5ex}

In this appendix we report, for the sake of completeness, some
auxiliary contents that, for space limitations, we could not include
in the paper. 

\section{Towards Extending Concolic Testing to Full Prolog} \label{fullsemantics}

\begin{figure}[b]
  \rule{\linewidth}{1pt}
  \[
  \hspace{-3ex}\begin{array}{r@{~}l}
    \mathsf{(success)} & {\displaystyle 
      \frac{~} 
        {\tuple{\mathsf{true}_\delta\midd S} \to \tuple{\success_\delta}}
        } 
        \\[2ex]

     \mathsf{(failure)} & {\displaystyle 
      \frac{~} 
        {\tuple{(\fail,\cB)_\delta} \to \tuple{\failsc_\delta}}
        }
        \hfill\hspace{5ex}
    \mathsf{(backtrack)} ~ {\displaystyle 
      \frac{S\neq\epsilon} 
        {\tuple{(\fail,\cB)\midd S} \to \tuple{S}}
        }\\[3ex]

     \textsf{(choice)} &  {\displaystyle 
      \frac{\defined(A) \wedge
        \clauses(A,\cP) = (c_1,\ldots,c_n)\wedge n>0 \wedge m~\mbox{is fresh}} 
        {\tuple{(A,\cB)_\delta\midd S} \to
          \tuple{(A,\cB)_\delta^{c_1[!/!^m]} \midd \ldots\midd
          (A,\cB)_\delta^{c_n[!/!^m]}\midd ~?^m_\delta\midd  S}}
        }\\[3ex]

    \mathsf{(choice\_fail)} & {\displaystyle 
      \frac{\defined(A,\cP)\wedge \clauses(A,\cP)=\emptyset} 
        {\tuple{(A,\cB)_\delta\midd S} \to \tuple{(\fail,\cB)_\delta\midd S}}
        }\\[3ex]

    \mathsf{(unfold)} & {\displaystyle 
      \frac{\mgu(A,H_1)=\sigma} 
        {\tuple{(A,\cB)_\delta^{H_1\leftarrow\cB_1}\midd S}
          \to \tuple{(\cB_1\sigma,\cB\sigma)_{\delta\sigma}\midd S}}
        } \\[3ex]

    \mathsf{(cut)} & {\displaystyle 
      \frac{~} 
        {\tuple{(!^m,\cB)_\delta\midd S'\midd ~?^m_{\delta'}\midd S}
          \to \tuple{\cB_{\delta}\midd~?^m_{\delta'}\midd S}}
        } 
        \hfill\hspace{5ex}
    \mathsf{(cut\_fail)} ~ {\displaystyle 
      \frac{~} 
        {\tuple{?^m_\delta \midd S}
          \to \tuple{\fail_\delta \midd S}}
        } \\[3ex]

    \mathsf{(call)} & {\displaystyle 
      \frac{A\not\in\cV\wedge m~\mbox{is fresh}} 
        {\tuple{(call(A),\cB)_\delta\midd S}
          \to
          \tuple{(A[\cV/\call(\cV),!/!^m],\cB)_{\delta}\midd ~?^m_{\delta}\midd S}}
        } \\[3ex]

    \mathsf{(call\_error)} & {\displaystyle 
      \frac{A\in\cV} 
        {\tuple{(call(A),\cB)_\delta\midd S}
          \to
          \tuple{\errorsc_{\delta}}}
        } \\[3ex]

    \mathsf{(not)} & {\displaystyle 
      \frac{m~\mbox{is fresh}} 
        {\tuple{(\backslash\!\!+\!\!(A),\cB)_\delta\midd S}
          \to
          \tuple{(\call(A),!^m,\fail)_{\delta}\midd \cB_\delta\midd ~?^m_\delta\midd S}}
        } \\[3ex]

    \mathsf{(unify)} & {\displaystyle 
      \frac{\mgu(t_1,t_2)=\sigma\neq\fail} 
        {\tuple{(t_1=t_2,\cB)_\delta\midd S}
          \to
          \tuple{\cB\sigma_{\delta\sigma}\midd S}}
        } 
        \hfill\hspace{3ex}
    \mathsf{(unify\_fail)} ~ {\displaystyle 
      \frac{\mgu(t_1,t_2)=\fail} 
        {\tuple{(t_1=t_2,\cB)_\delta\midd S}
          \to
          \tuple{\fail_{\delta}\midd S}}
        } \\[3ex]

    \mathsf{(is)} & {\displaystyle 
      \frac{\mathsf{eval}(e_2) = t_2 \neq \error} 
        {\tuple{(t_1~\mathsf{is}~e_2,\cB)_\delta\midd S}
          \to
          \tuple{(t_1=t_2,\cB)_\delta\midd ~ S}}
        } 
        \hfill\hspace{1ex}
    \mathsf{(is\_error)} ~ {\displaystyle 
      \frac{\mathsf{eval}(e_2) = \error} 
        {\tuple{(t_1~\mathsf{is}~e_2,\cB)_\delta\midd S}
          \to
          \tuple{\errorsc_\delta}}
        } \\[3ex]

    \mathsf{(rel)} & {\displaystyle 
      \frac{\mathsf{eval}(t_1\oplus t_2) = A \in\{\mathsf{true},\mathsf{fail}\}} 
        {\tuple{(t_1~{\oplus}~t_2,\cB)_\delta\midd S}
          \to
          \tuple{(A,\cB)_\delta\midd ~ S}}
        } 
        \hfill\hspace{1ex}
    \mathsf{(rel\_error)} ~ {\displaystyle 
      \frac{\mathsf{eval}(t_1\oplus t_2) = \error} 
        {\tuple{(t_1\oplus t_2,\cB)_\delta\midd S}
          \to
          \tuple{\errorsc_\delta}}
        }
    \end{array}
    \]
  \rule{\linewidth}{1pt}
  \caption{Extended concrete semantics} 
  \label{fig:concrete3}
\end{figure}

In this section, we show a summary of our preliminary research on
extending concolic execution to deal with full Prolog. First, we
consider the extension of the concrete semantics. Here, we mostly
follow the linear semantics of 
(Str{\"o}der et al. 2011),
being the main
differences that we consider built-ins explicitly, we excluded dynamic
predicates for simplicity |but could be added along the lines of
(Str{\"o}der et al. 2011)|
and that, analogously to what we did in
Section~\ref{concrete}, only the first answer for the initial goal is
considered.

In the following, we let the Boolean function $\mathsf{defined}$
return true when its argument is an atom rooted by a defined predicate
symbol, and false otherwise (i.e., a built-in). Moreover, for evaluating
relational and arithmetic expressions, we assume a function
$\mathsf{eval}$ such that, given an expression $e$, $\mathsf{eval}(e)$
either returns the evaluation of $e$ (typically a number or a Boolean
value) or the special constant $\error$ when the expression is not
instantiated enough to be evaluated. E.g., $\eval(2+2) = 4$,
$\eval(3>1) = \mathit{true}$, but $\eval(X>0) = \error$.

The transitions rules are shown in Figure~\ref{fig:concrete3}. In the
following, we briefly explain the novelties w.r.t.\ the rules of
Section~\ref{concrete}: 
\begin{itemize}
\item In rule \textsf{choice} we use the notation $c[!/!^m]$ to denote
  a copy of clause $c$ where the occurrences of (possibly labeled)
  cuts $!$ at predicate positions (e.g., not inside a \textsf{call}),
  if any, are replaced by a \emph{labeled} cut $!^m$, where $m$ is a
  fresh label.  Also, in the derived state, we add a \emph{scope
    delimiter} $?^m$.  

\item Rule \textsf{cut} removes some alternatives from the current
  state, while rule \textsf{cut\_fail} applies when a goal reaches the
  scope delimiter without success.

\item The rules for \textsf{call} and negation should be clear. Let us
  only mention that the notation $A[\cV/\textsf{call}(\cV),!/!^m]$
  denotes the atom $A$ in which all variables $X$ on predicate
  positions are replaced by $\call(X)$ and all (possibly labeled) cuts
  on predicate positions are replaced by $!^m$.

\item Calls to the built-in predicate $\mathsf{is}$ are dealt with
  rules \textsf{is} and \textsf{is\_error} by means of the auxiliary
  function $\eval$. Rules \textsf{rel} and \textsf{rel\_error} proceed
  analogously with relational operators like $>$, $<$, $==$, etc.
\end{itemize}

\begin{figure}[b]
  \footnotesize
  \rule{\linewidth}{1pt}
  \[
  \hspace{-3ex}\begin{array}{r@{~}l}
    \mathsf{(success)} & {\displaystyle 
      \frac{~} 
        {\tuple{\mathsf{true}_\delta\mid S
            \sep\mathsf{true}_{\theta}\mid S'} 
          \leadsto_\diamond
          \tuple{\success_\delta \sep \success_\theta}}
        } \\[3ex]

    \mathsf{(failure)} & {\displaystyle 
      \frac{~} 
        {\tuple{(\mathsf{fail},\cB)_\delta
          \sep (\mathsf{fail},\cB')_\theta} \leadsto_\diamond 
        \tuple{\fail_\delta \sep \fail_\theta}}
        }\\[3ex]

    \mathsf{(backtrack)} & {\displaystyle 
      \frac{S\neq\epsilon} 
        {\tuple{(\mathsf{fail},\cB)\mid S \sep
          (\mathsf{fail},\cB')\mid S'} \leadsto_\diamond \tuple{S\sep S'}}
        }\\[3ex]

     \mathsf{(choice)} &  {\displaystyle 
      \frac{\mathsf{defined}(A) \wedge
        \clauses(A,\cP) = \ol{c_n}\wedge n>0 \wedge m~\mbox{is fresh}
      \wedge \clauses(A',\cP)=\ol{d_k}} 
        {\begin{array}{l@{~}l}
            \tuple{(A,\cB)_\delta\mid S \sep (A',\cB')_\theta\mid S'} \\ 
            \hspace{15ex}\leadsto_{c(\ell(\ol{c_n}),\ell(\ol{d_k}))} 
          \tuple{(A,\cB)_\delta^{c_1[!/!^m]} \mid \ldots\mid
          (A,\cB)_\delta^{c_n[!/!^m]}\mid ~?^m_\delta\mid  S\\
        \hspace{33ex}\sep (A',\cB')_\theta^{c_1[!/!^m]} \mid \ldots\mid
          (A',\cB')_\theta^{c_n[!/!^m]}\mid ~?^m_\theta\mid  S'
        }\end{array}}
        }\\[11ex]

    \mathsf{(choice\_fail)} & {\displaystyle 
      \frac{\mathsf{defined}(A,\cP)\wedge \clauses(A,\cP)=\emptyset
      \wedge \clauses(A',\cP)=\ol{c_k}} 
        {\tuple{(A,\cB)_\delta\mid S \sep (A',\cB')_\theta\mid S'} 
          \leadsto_{c(\emptyset,\ell(\ol{c_k}))} 
          \tuple{(\fail,\cB)_\delta\mid S\sep (\fail,\cB')_\theta\mid S'}}
        }\\[4ex]

    \mathsf{(unfold)} & {\displaystyle 
      \frac{\mgu(A,H_1)=\sigma\wedge\mgu(A',H_1)=\sigma'} 
        {\tuple{(A,\cB)_\delta^{H_1\leftarrow\cB_1}\mid S\sep
          (A',\cB')_\theta^{H_1\leftarrow\cB_1}\mid S'}
          \leadsto_\diamond \tuple{(\cB_1\sigma,\cB\sigma)_{\delta\sigma}\mid S
          \sep (\cB_1\sigma',\cB'\sigma')_{\theta\sigma'}\mid S'}}
        } \\[3ex]

    \mathsf{(cut)} & {\displaystyle 
      \frac{~} 
        {\tuple{(!^m,\cB)_\delta\mid S_1\mid ~?^m_{\delta'}\mid S
          \sep (!^m,\cB')_\theta\mid S'_1\mid ~?^m_{\theta'}\mid S'}
          \leadsto_\diamond \tuple{\cB_{\delta}\mid~?^m_{\delta'}\mid S
          \sep \cB'_{\theta}\mid~?^m_{\theta'}\mid S'}}
        } \\[3ex]

    \mathsf{(cut\_fail)} & {\displaystyle 
      \frac{~} 
        {\tuple{?^m_\delta \mid S\sep ?^m_\theta \mid S'}
          \leadsto_\diamond \tuple{\mathsf{fail}_\delta \mid S\sep \mathsf{fail}_\theta \mid S'}}
        } \\[3ex]

    \mathsf{(call)} & {\displaystyle 
      \frac{A\not\in\cV\wedge m~\mbox{is fresh}} 
      {\begin{array}{l}
          \tuple{(call(A),\cB)_\delta\mid S
        \sep (call(A'),\cB')_\theta\mid S'} \\
          \hspace{10ex}\leadsto_\diamond
          \tuple{(A[\cV/\mathsf{call}(\cV),!/!^m],\cB)_{\delta}\mid ~?^m_{\delta}\mid S
          \sep (A'[\cV/\mathsf{call}(\cV),!/!^m],\cB')_{\theta}\mid ~?^m_{\theta}\mid S'}
      \end{array}}
        } \\[6ex]

    \mathsf{(call\_error)} & {\displaystyle 
      \frac{A\in\cV} 
        {\tuple{(call(A),\cB)_\delta\mid S \sep (call(A'),\cB')_\theta\mid S'}
          \leadsto_\diamond
          \tuple{\error_{\delta}\sep \error_{\theta}}}
        } \\[3ex]

    \mathsf{(not)} & {\displaystyle 
      \frac{m~\mbox{is fresh}} 
        {\begin{array}{l}
            \tuple{(\backslash\!\!+\!\!(A),\cB)_\delta\mid S
          \sep (\backslash\!\!+\!\!(A'),\cB')_\theta\mid S'} \\
          \hspace{10ex}\leadsto_\diamond
          \tuple{(\mathsf{call}(A),!^m,\mathsf{fail})_{\delta}\mid \cB_\delta\mid ~?^m_\delta\mid S
          \sep (\mathsf{call}(A'),!^m,\mathsf{fail})_{\theta}\mid \cB'_\theta\mid ~?^m_\theta\mid S'}
      \end{array}}
        } \\[6ex]

    \mathsf{(unify)} & {\displaystyle 
      \frac{\mgu(t_1,t_2)=\sigma \wedge \mgu(t'_1,t'_2)=\sigma'} 
        {\tuple{(t_1=t_2,\cB)_\delta\mid S\sep (t'_1=t'_2,\cB')_\theta\mid S'}
          \leadsto_{u(t'_1,t'_2)}
          \tuple{\cB\sigma_{\delta\sigma}\mid S\sep \cB'\sigma'_{\delta\sigma'}\mid S'}}
        } \\[3ex]

    \mathsf{(unify\_fail)} & {\displaystyle 
      \frac{\mgu(t_1,t_2)=\mathsf{fail}} 
        {\tuple{(t_1=t_2,\cB)_\delta\mid S\sep (t'_1=t'_2,\cB')_\theta\mid S'}
          \leadsto_{d(t'_1,t'_2)}
          \tuple{\mathsf{fail}_{\delta}\mid S\sep \mathsf{fail}_{\theta}\mid S'
          }}
        } \\[3ex]

    \mathsf{(is)} & {\displaystyle 
      \frac{\mathsf{eval}(e_2) = t_2 \neq \mathsf{error}
      \wedge \mathsf{sym\_eval}(e'_2) = t'_2\wedge \mbox{$X$ is fresh}} 
        {\tuple{(t_1~\mathsf{is}~e_2,\cB)_\delta\mid S
          \sep (t'_1~\mathsf{is}~e'_2,\cB')_\theta\mid S'}
          \leadsto_{is(X,t'_2)}
          \tuple{(t_1=t_2,\cB)_\delta\mid ~ S\sep
          (t'_1=X,\cB')_\theta\mid ~ S'}}
        } \\[3ex]

    \mathsf{(is\_error)} & {\displaystyle 
      \frac{\mathsf{eval}(e_2) = \mathsf{error}} 
        {\tuple{(t_1~\mathsf{is}~e_2,\cB)_\delta\mid S
          \sep (t'_1~\mathsf{is}~e'_2,\cB')_\theta\mid S'}
          \leadsto_\diamond
          \tuple{\error_\delta\sep \error_\theta}}
        } \\[3ex]

    \mathsf{(rel)} & {\displaystyle 
      \frac{\mathsf{eval}(t_1\oplus t_2) = A \in\{\mathsf{true},\mathsf{fail}\}
      \wedge \mathsf{sym\_eval}(t'_1\oplus t'_2) = A'} 
      {\tuple{(t_1~{\oplus}~t_2,\cB)_\delta\mid S\sep
          (t'_1~{\oplus}~t'_2,\cB')_\theta\mid S'}
          \leadsto_{r(A',A)}
          \tuple{(A,\cB)_\delta\mid ~ S\sep (A',\cB')_\theta\mid ~ S'}}
        } \\[3ex]

    \mathsf{(rel\_error)} & {\displaystyle 
      \frac{\mathsf{eval}(t_1\oplus t_2) = \mathsf{error}} 
        {\tuple{(t_1\oplus t_2,\cB)_\delta\mid S\sep (t'_1\oplus t'_2,\cB')_\theta\mid S'}
          \leadsto_\diamond
          \tuple{\error_\delta\sep \error_\theta}}
        } \\[3ex]

    \end{array}
    \]
  \rule{\linewidth}{1pt}
  \caption{Extended concolic execution semantics} \label{fig:concolic3}
\end{figure}

\noindent
Regarding the concolic execution semantics, we follow a similar
approach to that of Section~\ref{concolic}. The labeled transition
rules can be seen in Figure~\ref{fig:concolic3}. Now, we consider six
kinds of labels for $\leadsto$:
\begin{itemize}
\item The labels $\diamond$ and $c(\cL_1,\cL_2)$ with the same meaning
  as in the concolic semantics of Section~\ref{concolic}.

\item The label $u(t_1,t_2)$, which is used to denote a unification
  step, i.e., the step implies that $t_1$ and $t_2$ should unify.

\item In contrast, the label $d(t_1,t_2)$ denotes a disunification,
  i.e., the step implies that $t_1$ and $t_2$ should not unify.

\item The label $is(X,t)$ denotes a step where $\mathsf{is}$ is
  evaluated (see below). 

\item Finally, the label $r(A',A)$ denotes that the relational
  expression $A'$ should be equal to $A\in\{\true,\fail\}$.
\end{itemize}
In particular, in rules \textsf{unify} and \textsf{unify\_fail}, the
labels store the unification that must hold in the step. Note that the
fact that $\mgu(t_1,t_2)=\fail$ does not imply $\mgu(t'_1,t'_2)=\fail$
since $t'_1$ and $t'_2$ might be less instantiated than $t_1$ and
$t_2$.

In rule \textsf{is}, we label the step with $is(X,t'_2)$ which means
that the fresh variable $X$ should be bound to the evaluation of
$t'_2$ after grounding it. Note that introducing such a fresh variable
is required to avoid a failure in the subsequent step with rule
\textsf{unify} because of, e.g., a non-ground arithmetic expression
that could not be evaluated yet to a value using function
$\mathsf{sym\_eval}$.
Note that rule \textsf{is\_error} does not include any label since we
assume that an error in the concrete computation just aborts the
execution and also the test case generation process.

Finally, in rule \textsf{rel} we label the step with $r(A',A)$ where
$A$ is the value $\true/\fail$ of the relational expression in the
concrete goal, and $A'$ is a (possibly nonground) corresponding
expression in the symbolic goal. Here, we use the auxiliary function
$\mathsf{sym\_eval}$ to simplify the relational expression as much as
possible. E.g., $\mathsf{sym\_eval}(3>0) = \true$ but
$\mathsf{sym\_eval}(3+2>X) = 5>X$.

These labels can be used for extending the concolic testing algorithm
of Section~\ref{testing}. For instance, given a concolic execution
step labeled with $r(X>0,\true)$, we have that solving $\neg(X>0)$
will produce a binding for $X$ (e.g., $\{X/0\}$) that will follow an
alternative path. Here, the concolic testing procedure will integrate
a constraint solver for producing solutions to negated constraints. We
find this extension of the concolic testing procedure an interesting
topic for future work.

\section{Proofs of Technical Results} \label{appendix:proofs}

\subsection{Concolic Execution Semantics}

\begin{proof}[Proof of Theorem~\ref{th:invariant}]
  Since the base case $i=0$ trivially holds, in the following we only
  consider the inductive case $i>0$. Let $C_i=
  \tuple{\cB^c_\delta\midd S\sep \cD^{c'}_\theta\midd S'}$. By the
  inductive hypothesis, we have $|S|=|S'|$, $\cD\sleq \cB$, $c=c'$ (if
  any), and $p(\ol{X_n})\theta \sleq p(\ol{t_n})\delta$. Now, we
  consider the step $C_i \leadsto C_{i+1}$ and distinguish the
  following cases, depending on the applied rule:
  \begin{itemize}
  \item If the rule applied is $\mathsf{success}$, $\mathsf{failure}$,
    $\mathsf{backtrack}$ or $\mathsf{choice\_fail}$, the claim follows
    trivially by induction.

  \item If the rule applied is $\mathsf{choice}$, let us assume that
    we have $\cB = (A,\cB')$, $\cD = (A',\cD')$ and $\clauses(A,\cP) =
    \ol{c_j}$, $j>0$. Therefore, we have $C_{i+1}=
    \tuple{\cB^{c_1}_\delta\midd\ldots\midd \cB^{c_j}_\delta\midd
      S\sep \cD^{c_1}_\theta\midd\ldots\midd \cD^{c_j}_\theta\midd
      S'}$, and the claim follows straightforwardly by the induction
    hypothesis.

  \item Finally, if the applied rule is $\mathsf{unfold}$, then we
    have that $\cB^c_\delta = (A,\cB')^c_\delta$, $\cD^c_\theta =
    (A',\cD')^c_\theta$ for some clause $c =
    H_1\leftarrow\cB_1$. Therefore, we have $C_{i+1} =
    \tuple{(\cB_1\sigma,\cB'\sigma)_{\delta\sigma}\midd S\sep
      (\cB_1\sigma',\cD'\sigma')_{\theta\sigma'}\midd S'}$, where
    $\mgu(A,H_1)=\sigma$ and $\mgu(A',H_1) = \sigma'$. First, $c=c'$
    holds by vacuity since the goals are not labeled with a
    clause. Also, the number of concrete and symbolic goals is
    trivially the same since $|S|=|S'|$ by the inductive hypothesis.
    Now, by the inductive hypothesis, we have $\cD\sleq\cB$ and thus
    $A'\sleq A$ and $\cD'\sleq\cB'$. Then, since $\sigma=\mgu(A,H_1)$,
    $\sigma'=\mgu(A',H_1)$, $\var(H_1\leftarrow
    \cB_1)\cap\var(A)=\emptyset$, and $\var(H_1\leftarrow
    \cB_1)\cap\var(A')=\emptyset$, it is easy to see that
    $A'\sigma'\sleq A\sigma$ (and thus $\cD'\sigma'\sleq\cB'\sigma$)
    and $\sigma'\sleq \sigma$ when restricted to the variables of $H_1$ 
    (and thus $\cB_1\sigma'\sleq
    \cB_1\sigma$). Therefore, we can conclude
    $(\cB_1\sigma',\cD'\sigma')\sleq (\cB_1\sigma,\cB'\sigma)$.
    Finally, using a similar argument, we have
    $p(\ol{X_n})\theta\sigma'\sleq p(\ol{t_n})\delta\sigma$.
  \end{itemize}
\end{proof}

\subsection{Solving Unifiability Problems}

First, we prove the following invariant which justifies that the
algorithm in Definition~\ref{alg1} is well defined.

\begin{proposition}
  The following statement is an invariant of the loops at
  lines~\ref{algo-msa-while-simple} and
  \ref{algo-msa-while-not-simple} of the algorithm in
  Definition~\ref{alg1}:
  \begin{description}
  \item[$\mathrm{(invariant)}$] (a) $A\approx B$ for all $B\in \cB$
    and (b) $A\leq B'$ for some $B'\in\cB$.
  \end{description}
\end{proposition}

\begin{proof}
  Let us first consider the loop at line
  \ref{algo-msa-while-simple}. Clearly, the invariant holds upon
  initialization. Therefore, let us assume that it holds for some
  arbitrary set $\cB$ and we prove it also holds for $\cB' = \cB\eta$
  with $\eta = \{X/t\}$ for some simple disagreement pair $X,t$ (or
  $t,X$).  
  Let us consider part (a).  Since $A\approx B$ for all $B\in\cB$,
  there exist a substitution $\theta$ such that $A\theta=B\theta$ for
  all $B\in\cB$. Consider such an arbitrary $B\in\cB$. If
  $X\not\in\Var(B)$, then part (a) of the invariant holds trivially in
  $\cB'$. Otherwise, $\theta\{X/t\}$ is clearly a unifier $A$ and $B$,
  and it also holds. 
  Consider now part (b). Since $A\leq B'$ for some $B'\in\cB$, there
  exists a substitution $\sigma$ such that $A\sigma=B'$. Using a
  similar argument as before, either $A\sigma=B'$ with $B'\in\cB'$ or
  $A\sigma\{X/t\}=B'\{X/t\}$ with $B'\{X/t\}\in\cB$, and part (b) of
  the invariant also holds in $\cB'$.

  Let us now consider the loop at line
  \ref{algo-msa-while-not-simple}. Clearly, the invariant holds when
  the previous loop terminates. 
  Let $t,t'$ be the selected disagreement pair.
  Then $t,t'$ is replaced in $\cB$ by a fresh variable
  $U\in \cU$, thus obtainining a new set $\cB'$.
  Let $\eta_1:=\{U/t\}$ and $\eta_2:=\{U/t'\}$. Both $\eta_1$ and $\eta_2$
  are idempotent substitutions because $U\not\in \Var(t)$ and $U\not\in
  \Var(t')$ since $U$ is fresh.
  Let $B_1,B_2$ be the atoms of $\cB$ where $t,t'$ come from
  and $C_1,C_2$ be the atoms obtained by replacing
  $t,t'$ in $B_1,B_2$ by $U$. Then $B_1=C_1\eta_1$ and
  $B_2=C_2\eta_2$. 
  Now, we want to prove that the invariant also holds in $\cB' =
  \cB\setminus\{B_1,B_2\}\cup\{C_1,C_2\}$. Part (a) is trivial, since
  we only generalize some atoms: if $A$ unify with $B_1$ and $B_2$, it
  will also unify with $C_1$ and $C_2$. Regarding part (b), we have
  that $A\leq B'$ for some $B'\in\cB$. Clearly, part (b) also holds in
  $\cB'$ if $B'$ is different from $B_1$ and $B_2$. Otherwise,
  w.l.o.g., assume that $B'=B_1$ and $A\leq B_1$. Since $A\approx B_1$
  and $A\approx B_2$, and $t,t'$ is a disagreement pair for $B_1,B_2$,
  we have that the subterm of $A$ that corresponds to the position of
  $t,t'$ should be more general than $t,t'$ (otherwise, it would not
  unify with both terms). Therefore, replacing $t$ by a fresh variable
  $U$ will not change that, and we have $A\leq C_1$ for some
  $C_1\in\cB$.
\end{proof}
The following auxiliary results are useful to prove the correctness of
the algorithms 
in Definitions~\ref{alg1} 
and \ref{alg2}. 

\begin{lemma}\label{lemma:technical-1}
  Suppose that $A\theta=B\theta$ for some atoms $A$ and $B$
  and some substitution $\theta$.
  Then we have $A\theta\eta=B\eta\theta\eta$ for any
  substitution $\eta$ with
  $[\Dom(\eta)\cap \Var(B)]\cap\Dom(\theta)=\emptyset$
  and
  $Ran(\eta)\cap\Dom(\theta\eta)=\emptyset$.
\end{lemma}
\begin{proof}
  For any $X\in\Var(B)$, 
  \begin{itemize}
  \item either $X\not\in\Dom(\eta)$ and then 
    $X\eta\theta\eta=X\theta\eta$
  \item or $X\in\Dom(\eta)$ and then
    $X\eta\theta\eta=(X\eta)\theta\eta=X\eta$
    because $\Ran(\eta)\cap\Dom(\theta\eta)=\emptyset$.
    Moreover, $X\not\in\Dom(\theta)$ because
    $[\Dom(\eta)\cap \Var(B)]\cap\Dom(\theta)=\emptyset$,
    so $X\theta\eta=X\eta$.
    Finally, $X\eta\theta\eta=X\theta\eta$.
  \end{itemize}
  Consequently, $B\eta\theta\eta=B\theta\eta$.
  As $A\theta=B\theta$, we have $A\theta\eta=B\theta\eta$
  i.e. $A\theta\eta=B\eta\theta\eta$.
\end{proof}

\begin{proposition}\label{proposition:invariant-correction-algo-pos-1}
  The loop at line~\ref{algo-msa-while-simple} always terminates
  and the following statement is an invariant of this loop.
  \begin{description}
  \item[$\mathrm{(inv)}$] For each $A'\in\{A\}\cup\Hpos$ there exists 
    $B\in\cB$ and a substitution $\theta$ such that
    $A'\theta=B\theta$ and
    $\Var(\cB)\cap\Dom(\theta)=\emptyset$. 
  \end{description}
\end{proposition}
\begin{proof}
  Action~(\ref{algo-msa-simple-pair})
  reduces the number of simple disagreement pairs in
  $\cB$ which implies termination of the 
  loop at line~\ref{algo-msa-while-simple}.

  Let us prove that $\mathrm{(inv)}$ is an invariant.
  First, $\mathrm{(inv)}$ clearly holds upon initialization
  of $\cB$. Suppose it holds prior to an execution of
  action~(\ref{algo-msa-simple-pair}).
  Therefore, for each $A'\in\{A\}\cup\Hpos$ there exists 
    $B\in\cB$ and a substitution $\theta$ such that
    $A'\theta=B\theta$ and
    $\Var(\cB)\cap\Dom(\theta)=\emptyset$.
  Let $t,t'$ be the selected simple disagreement pair.
  Then, we consider a substitution $\eta$
  determined by $t,t'$. For any $X\in\Ran(\eta)$,
  we have $X\in\Var(\cB)$. 
  Thus $X\not\in\Dom(\theta)$ by $\mathrm{(inv)}$.
  Hence
  $\Ran(\eta)\cap\Dom(\theta)=\emptyset$.
  Moreover, as $t,t'$ is a simple pair we have
  $\Ran(\eta)\cap\Dom(\eta)=\emptyset$. Hence,
  \begin{equation}\label{eq:ran-eta}
    \Ran(\eta)\cap\Dom(\theta\eta)=\emptyset\;.
  \end{equation}
  Since $B\in\cB$, we have $[\Dom(\eta)\cap \Var(B)]\cap\Dom(\theta)=\emptyset$.
  Consequently, by~(\ref{eq:ran-eta}) and
  Lemma~\ref{lemma:technical-1} we have
  \[A'\theta\eta=B\eta\theta\eta\;.\]
  Now, we want to prove that $\mathrm{(inv)}$ holds for $\cB\eta$,
  i.e., that for each $A'\in\{A\}\cup\Hpos$ there exists 
    $B\eta\in\cB\eta$ and a substitution $\theta'$ such that
    $A'\theta'=B\eta\theta'$ and
    $\Var(\cB\eta)\cap\Dom(\theta')=\emptyset$.
  We let $\theta'=\theta\eta$, so $A'\theta\eta=B\eta\theta\eta$ holds.
  Now, suppose by contradiction that
  $\Var(\cB\eta)\cap\Dom(\theta\eta)\neq\emptyset$, and let $X$ be one of its elements.
  We have $X\not\in\Dom(\eta)$ because
  $\Ran(\eta)\cap\Dom(\eta)=\emptyset$, so
  $X\in\Dom(\theta)$.
  Moreover, $X\not\in\Ran(\eta)$ by~(\ref{eq:ran-eta}) so
  $X\in\Var(\cB)$. Therefore,
  $X\in\Var(\cB)\cap\Dom(\theta)$
  which by $\mathrm{(inv)}$ gives a contradiction.
  Consequently, 
  \[\Var(\cB\eta)\cap\Dom(\theta\eta)=\emptyset\]
  and the claim follows.
\end{proof}

\begin{proposition}\label{proposition:invariant-correction-algo-pos-2}
  The loop at line~\ref{algo-msa-while-not-simple} always terminates
  and the following statement is an invariant of this loop.
  \begin{description}
  \item[$\mathrm{(inv')}$] For each $A'\in\{A\}\cup\Hpos$ there exists 
    $B\in\cB$ and a substitution $\theta$ such that
    $A'\theta=B\theta$ and
    $\Var(\cB)\cap\Dom(\theta)
    \subseteq \cU$.
  \end{description}
\end{proposition}
\begin{proof}
  Action~(\ref{algo-msa-not-simple-pair})
  reduces the number of disagreement pairs in
  $\cB$ which implies termination of the 
  loop at line~\ref{algo-msa-while-not-simple}.

  Let us prove that $\mathrm{(inv')}$ is an invariant.
  By Proposition~\ref{proposition:invariant-correction-algo-pos-1},
  $\mathrm{(inv)}$ holds upon termination of the loop at
  line~\ref{algo-msa-while-simple}, hence $\mathrm{(inv')}$ holds
  just before execution of the loop at
  line~\ref{algo-msa-while-not-simple}.
  Suppose it holds prior to an execution of
  action~(\ref{algo-msa-not-simple-pair}), so we have that, 
  for each $A'\in\{A\}\cup\Hpos$ there exists 
    $B\in\cB$ and a substitution $\theta$ such that
    $A'\theta=B\theta$ and
    $\Var(\cB)\cap\Dom(\theta)
    \subseteq \cU$.
  Let $t,t'$ be the selected disagreement pair.
  Then $t,t'$ is replaced in $\cB$ by a fresh variable
  $U\in \cU$, thus obtainining a new set $\cB'$.
  Let $\eta_1:=\{U/t\}$ and $\eta_2:=\{U/t'\}$. Both $\eta_1$ and $\eta_2$
  are idempotent substitutions because $U\not\in \Var(t)$ and $U\not\in
  \Var(t')$ since $U$ is fresh.
  Let $B_1,B_2$ be the atoms of $\cB$ where $t,t'$ come from
  and $C_1,C_2$ be the atoms obtained by replacing
  $t,t'$ in $B_1,B_2$ by $U$. Then $B_1=C_1\eta_1$ and
  $B_2=C_2\eta_2$. 
  Now, we want to prove that $\mathrm{(inv')}$ holds in $\cB' =
  \cB\setminus\{B_1,B_2\}\cup\{C_1,C_2\}$, i.e., that for each
  $A'\in\{A\}\cup\Hpos$ there exists $B\in\cB'$ and a substitution
  $\theta$ such that $A'\theta=B\theta$ and
  $\Var(\cB')\cap\Dom(\theta) \subseteq \cU$.
  
  Since $\mathrm{(inv')}$ holds in $\cB$, we have $A'\theta=B\theta$.
  Moreover, $A'=A'\eta_1=A'\eta_2$ because $U$ does not occur in $A'$.
  So if $B=B_1$ then $A'\eta_1\theta=C_1\eta_1\theta$ and
  if $B=B_2$ then $A'\eta_2\theta=C_2\eta_2\theta$.
  Consequently, let us set
  \begin{itemize}
  \item $\theta':=\theta$ and $B':=B$ if $B\not\in\{B_1,B_2\}$
  \item $\theta':=\eta_1\theta$ and $B':=C_1$ if $B=B_1$
  \item $\theta':=\eta_2\theta$ and $B':=C_2$ if $B=B_2$.
  \end{itemize}
  Then we have
  \begin{equation}\label{eq:A-unifies}
    A'\theta'=B'\theta'\;.
  \end{equation}
  Moreover,
  $\Dom(\theta')\subseteq
  \Dom(\theta)\cup\Dom(\eta_1)\cup\Dom(\eta_2)$
  i.e.
  \begin{equation}\label{eq:dom-theta-prime}
    \Dom(\theta')\subseteq
    \Dom(\theta)\cup\{U\}\;.
  \end{equation}
  As $\Var(C_1,C_2)\subseteq\Var(B_1,B_2)\cup\{U\}$ then
  \[\Var(C_2,C_2)\cap\Dom(\theta')
  \subseteq \cU\]
  because 
  $\Var(B_1,B_2)\cap\Dom(\theta)\subseteq \cU$
  by $\mathrm{(inv')}$ and 
  $\Var(B_1,B_2)\cap\{U\}=
  \{U\}\cap\Dom(\theta) = \emptyset$ and
  $\{U\}\cap\{U\}\subseteq \cU$.
  Moreover, by $\mathrm{(inv')}$ we have
  $\Var(\cB)\cap(\Dom(\theta)\cup\{U\})
  \subseteq \cU$ so by~(\ref{eq:dom-theta-prime})
  \[\Var(\cB)\cap\Dom(\theta')
  \subseteq \cU\;.\]
  Hence, $\Var(\cB\setminus\{B_1,B_2\}\cup\{C_1,C_2\})
  \cap\Dom(\theta')\subseteq \cU$.
  With~(\ref{eq:A-unifies}) this implies that upon termination of
  action~(\ref{algo-msa-not-simple-pair}) the invariant
  $\mathrm{(inv')}$ holds
  because $B_1$ is set to $C_1$ and $B_2$ to $C_2$.
\end{proof}
The correctness of the algorithm in Definition~\ref{alg1} is then
stated as follows.

\begin{theorem}\label{theorem:correction-algo-pos}
  Let $A$ be an atom and $\Hpos$ be a set of atoms such that
  $\var(\{A\}\cup\Hpos)\cap\cU=\emptyset$ and $A\approx B$ for all
  $B\in\Hpos$. 
  The algorithm in Definition~\ref{alg1} with input $A$ and $\Hpos$
  always terminates and returns a substitution $\theta$ such that
  $A\theta\eta$ unifies with all the atoms of $\Hpos$ for any
  idempotent substitution $\eta$ with
  $\Dom(\eta)\subseteq\Var(A\theta)$ and
  $\Var(\eta)\cap\cU=\emptyset$.
\end{theorem}

\begin{proof} 
  Proposition~\ref{proposition:invariant-correction-algo-pos-1}
  and Proposition~\ref{proposition:invariant-correction-algo-pos-2}
  imply termination of the algorithm.
  Upon termination of the loop at line~\ref{algo-msa-while-not-simple}
  we have $|\cB|=1$.
  Let $B$ be the element of $\cB$ with $A\theta=B$. 
  Now, we want to prove that $A\theta\eta$ unifies with all the
  atoms in $\Hpos$ for any idempotent substitution $\eta$ (i.e.,
  $\Dom(\eta)\cap\Ran(\eta)=\emptyset$) such that
  $\Dom(\eta)\subseteq\Var(A\theta)=\Var(B)$
  and $\Var(\eta)\cap\cU=\emptyset$.
  By Proposition~\ref{proposition:invariant-correction-algo-pos-2}, we
  have that, for all $B'\in\Hpos$, there exists a substitution
  $\theta'$ such that $B\theta'=B'\theta'$ and
  $\Var(B)\cap\Dom(\theta')\subseteq \cU$. 
  From all the previous conditions, it follows that
  $[\Dom(\eta)\cap\Var(B)]\cap\Dom(\theta')=\emptyset$ and
  $\Ran(\eta)\cap\Dom(\theta'\eta)=\emptyset$. Therefore, by
  Lemma~\ref{lemma:technical-1}, we have
  $B\eta\theta'\eta=B'\theta'\eta$.  Finally, since
  $A\theta=B$, we have $A\theta\eta\theta'\eta=B'\theta'\eta$
  and, thus, $A\theta\eta$ unifies with $B'$.
\end{proof}

\begin{proof}[Proof of Theorem~\ref{theorem:correction-algo-posneg}]
  Each step of the algorithm terminates, hence the algorithm terminates.
  Assume that the algorithm returns a substitution $\sigma$.
  The set $G\sigma$ is ground by construction.
  By Theorem~\ref{theorem:correction-algo-pos}, we have that $A\sigma
  = A\theta\eta$ unifies with all the atoms in $\Hpos$ as long as
  $\eta$ is idempotent, $\Dom(\eta)\subseteq\Var(A\theta)$ and
  $\Var(\eta)\cap\cU=\emptyset$. Finally, the last check ensures that
  $A\sigma$ does not unify with any atom of $\Hneg$.
\end{proof}

\subsubsection{Completeness} \label{sec:completeness}

For simplicity, we ignore the groundness constraint in this
section. Therefore, we now focus on the completeness of the following
unification problem: Let $A$ be an atom and $\Hpos,\Hneg$ be sets of
atoms such that $A\approx B$ for all $B\in\Hpos\cup\Hneg$. Then, we
want to find a substitution $\sigma$ such that
\[
A\sigma\approx B ~\mbox{for all}~B\in\Hpos 
~\mbox{but}~ \neg(A\sigma\approx B') ~\mbox{for all}~ B'\in\Hneg
\hfill(**)
\]
We further assume that all atoms are renamed apart. 

Let us first formalize the notion of unifying substitution:

\begin{definition}[unifying substitution]
  Let $A$ be an atom and let $\cB$ be a set of atoms such that
  $\var(A,\cB)\cap\cU=\emptyset$ and $A\approx B$ for all
  $B\in\cB$. We say that $\sigma$ is a unifying substitution for $A$
  w.r.t.\ $\cB$ if $A\sigma\approx B$ for all $B\in\cB$.
\end{definition}
In particular, we are interested in \emph{maximal} unifying
substitutions computed by the algorithm in Definition~\ref{alg1}.
The relevance of maximal unifying substitutions is that variables from
$\cU$ identify where further instantiation would result in a
substitution which is not a unifying substitution anymore.
For the remaining positions, we basically return their most general
unifier.

Now, we prove that binding an atom $A$ with a maximal unifying
substitution for $A$ w.r.t.\ $\Hpos$ does not affect to the existence
of a solution to our unification problem (**) above.
Here, for simplicity, we assume that only \emph{most specific}
solutions are considered, where a solution $\sigma$ is called a
\emph{most specific} solution for $A$ and $\Hpos,\Hneg$ if there
exists no other solution which is strictly less general than $\sigma$.
Furthermore, we also assume that the atom $A$ has the form
$p(X_1,\ldots,X_n)$.

\begin{lemma}
  Let $A$ be an atom and $\Hpos,\Hneg$ be sets of atoms such that
  $A\approx B$ for all $B\in\Hpos\cup\Hneg$. If there exists a
  substitution $\sigma$ such that $A\sigma\approx B$ for all
  $B\in\Hpos$ and $\neg(A\sigma\approx B)$ for all $B\in\Hneg$, then
  there exists a maximal unifying substitution $\theta$ and a
  substitution $\sigma'$ such that $A\theta\sigma'\approx B$ for all
  $B\in\Hpos$ and $\neg(A\theta\sigma'\approx B)$ for all $B\in\Hneg$.
\end{lemma}

\begin{proof} (\emph{sketch})
  Let us consider the stages of the algorithm in Definition~\ref{alg1}
  with input $\Hpos$ (atom $A$ is not needed since it has the form
  $p(X_1,\ldots,X_n)$ and, thus, imposes no constraint).  The first
  stage just propagates simple disagreement pairs of the form $X,t$ or
  $t,X$. When $X$ only occurs once, it is easy to see that $\sigma$ is
  also a (most specific) unifying substitution for $A$ w.r.t.\
  $\Hpos\{X/t\}$. Consider, e.g., that $\sigma$ contains a binding of
  the form $X_i/C[t']$ for some $i\in\{1,\ldots,n\}$ and context
  $C[~]$ and such that $t'$ corresponds to the same position of $X$
  and $t$ in $\Hpos$. Depending on the terms in the corresponding
  position of the remaining atoms, we might have $t'\leq t$ or $t\leq
  t'$. Either case, replacing $X$ by $t$ will not change the fact that
  $\sigma$ is still a most specific unifying substitution for
  $\Hpos\{X/t\}$.

  The step is more subtle when there are several simple disagreement
  pairs for a given variable, e.g., $X,t_1$ and $X,t_2$ (we could
  generalize it to an arbitrary number of pairs, but two are enough to
  illustrate how to proceed). In this case, if $t_1\leq t_2$, we
  choose $X,t_2$ and the reasoning is analogous to the previous case.
  However, when neither $t_1\leq t_2$ not $t_2\leq t_1$, the algorithm
  in Definition~\ref{alg1} is non-deterministic and allows us to
  choose any of them. As before, let us consider that $\sigma$
  contains bindings of the form $X_i/C[t'_1]$ and $X_j/C'[t'_2]$ for
  some $i,j\in\{1,\ldots,n\}$ and contexts $C[~],C'[~]$ and such that
  $t'_1$ and $t'_2$ correspond to the same positions of $t_1$ and
  $t_2$ in $\Hpos$, respectively. Here, assuming there are no further
  constraints from the remaining atoms, a most specific unifying
  substitution might either bind $X_i$ to $C[t_1]$ and leave $X_j$
  unconstrained (e.g., bound to a fresh variable) or the other way
  around: bind $X_i$ to $C[t_2]$ and leave $X_j$ unconstrained. Here,
  we choose the same alternative as in the considered solution
  $\sigma$, say $X_i$ is bound to $C[t_1]$. Therefore, $\sigma$ is
  still a unifying substitution for $A$ w.r.t.\ $\Hpos\{X/t_1\}$. Note
  that the new (non-simple) disagreement pair $t_1,t_2$ introduced in
  $\Hpos\{X/t_1\}$ will be generalized away in the next stage (and
  replaced by a fresh variable from $\cU$).
  
  Therefore, when the first stage is completed (i.e., step 2 in
  Definition~\ref{alg1}), we have propagated some terms from one atom to
  the remaining ones --as in the computation of a most general
  unifier-- thus producing a new set $\Hpos'$ such that $\sigma$ is
  still a (most specific) unifying substitution for $A$ w.r.t.\
  $\Hpos'$.

  By definition, after this stage, there are no simple disagreement
  pairs in $\Hpos'$. 
  Then, in the second stage (step 3 in Definition~\ref{alg1}), we replace
  every (non-simple) disagreement pair $t_1,t_2$ by a fresh variable
  $U$ from $\cU$. Since $\sigma$ was a unifying substitution for
  $\Hpos'$, it should have a binding $X_i/C[W]$ for some
  $i\in\{1,\ldots,n\}$ and context $C[~]$ and such that $W$
  corresponds to the same position of $t_1$ and $t_2$ in $\Hpos'$,
  where $W$ is a variable.  Therefore, replacing $t,t'$ by a fresh
  variable $U$ will not change the fact that $\sigma$ is still a
  unifying substitution for the resulting set (up to variable
  renaming).

  Hence, when the second stage is finished, we have a new set
  $\Hpos''$ without any disagreement pair at all, i.e.,
  $\Hpos''=\{B\}$ with $A\theta=B$. Moreover, since $\sigma$ is a most
  specific uniyfing substitution for $A$ w.r.t.\ $\Hpos''$, we have
  $\theta\leq\sigma~[\var(A)]$. Therefore, there exists a substitution
  $\sigma'$ such that $A\sigma=A\theta\sigma'$ such that $\sigma'$ is
  a solution for $A\theta$ and $\Hpos,\Hneg$, which concludes the
  proof.
\end{proof}

\section{Some More Examples on Solving Unifiability Problems}

\begin{example}[maximal unifying substitution]
  Let $A = p(X,Y)$ and $\Hpos = \{p(s(a),s(c)),p(s(b),s(c)),p(Z,Z)\}$.
  First the algorithm of Definition~\ref{alg1} sets
  $\cB := \{p(X,Y),p(s(a),s(c)),p(s(b),s(c)),p(Z,Z)\}$,
  then it considers the simple disagreement pairs
  in $\cB$. The substitution $\eta_1 := \{X/s(a)\}$
  is determined by $X,s(a)$. Action~(\ref{algo-msa-simple-pair})
  sets $\cB$ to $\cB\eta_1$ i.e. to
  \[\{p(s(a),Y),p(s(a),s(c)),p(s(b),s(c)),p(Z,Z)\}\;.\]
  The substitution $\eta_2 := \{Y/s(c)\}$
  is determined by $Y,s(c)$. Action~(\ref{algo-msa-simple-pair})
  sets $\cB$ to $\cB\eta_2=\{p(s(a),s(c)),p(s(b),s(c)),p(Z,Z)\}$.
  The substitution $\eta_3 := \{Z/s(c)\}$ is determined by
  $Z,s(c)$. Action~(\ref{algo-msa-simple-pair}) sets $\cB$ to
  $\cB\eta_3$ i.e. to
  \[\{p(s(a),s(c)),p(s(b),s(c)),p(s(c),s(c))\}\;.\]
  Now no simple disagreement pair occurs in $\cB$ hence the algorithm
  skips to the loop at line~\ref{algo-msa-while-not-simple}.
  \begin{itemize}
  \item Action~(\ref{algo-msa-not-simple-pair}) replaces 
    the disagreement pair $a,b$ with a
    fresh variable $U\in \cU$, hence $\cB$ is set to
    $\{p(s(U),s(c)),p(s(c),s(c))\}$.
  \item Action~(\ref{algo-msa-not-simple-pair}) replaces 
    the disagreement pair $U,c$
    with a fresh variable $U'\in \cU$, hence $\cB$ is set to
    $\{p(s(U'),s(c))\}$.
  \end{itemize}
  As $|\cB|=1$ the loop at line~\ref{algo-msa-while-not-simple}
  stops and the algorithm returns the substitution $\{X/s(U'),Y/s(c)\}$.

  Note that there are several non-deterministic possibilities
  for $\eta_1$, $\eta_2$ and $\eta_3$. For instance, if we
  consider $\eta_3:=\{Z/s(a)\}$, which is determined by $Z/s(a)$,
  then $\cB$ is set to $\{p(s(a),s(c)),p(s(b),s(c)),p(s(a),s(a))\}$.
  The loop at line~\ref{algo-msa-while-not-simple} finally sets 
  $\cB$ to $\{p(s(U),s(U'))\}$, so the algorithm returns the substitution
  $\{X/s(U),Y/s(U')\}$.
\end{example}
We note that the set $\mathcal{B}$ used by the algorithm
of Definition~\ref{alg1} may contain several occurrences
of a same, non-simple, disagreement pair.

\begin{example}[maximal unifying substitution]
  Let $A = p(X,Y)$ and $\Hpos = \{p(a,a),p(b,b)\}$.
  First the algorithm sets $\cB:=\{p(X,Y),p(a,a),p(b,b)\}$.
  Then the loop at line~\ref{algo-msa-while-simple}
  considers the simple disagreement pairs in $\cB$ and,
  for instance, it sets $\cB$ to $\{p(a,a),p(b,b)\}$
  (it may also set $\cB$ to $\{p(a,b),p(a,a),p(b,b)\}$
  or to $\{p(b,a),p(a,a),p(b,b)\}$).
  As no simple disagreement pair now occurs in $\cB$, the
  algorithm jumps at line~\ref{algo-msa-while-not-simple}.
  The pair $a,b$ occurs twice in $\mathcal{A}$.
  Action~(\ref{algo-msa-not-simple-pair}) replaces each occurrence
  with the same variable $U \in \cU$, so the loop at
  line~\ref{algo-msa-while-not-simple} sets
  $\cB$ to $\{p(U,U)\}$ and the algorithm returns $\{X/U,Y/U\}$.
\end{example}

\begin{example}[maximal unifying substitution]
  Let $A = p(X,Y)$ and $\Hpos = \{p(a,b),p(b,a)\}$.
  First the algorithm sets $\cB:=\{p(X,Y),p(a,b),p(b,a)\}$.
  Then the loop at line~\ref{algo-msa-while-simple}
  considers the simple disagreement pairs in $\cB$ and,
  for instance, it sets $\cB$ to $\{p(a,b),p(b,a)\}$
  (it may also set $\cB$ to $\{p(a,a),p(a,b),p(b,a)\}$
  or to $\{p(b,b),p(a,b),p(b,a)\}$).
  As no simple disagreement pair now occurs in $\cB$, the
  algorithm jumps at line~\ref{algo-msa-while-not-simple}.
  The pairs $a,b$ and $b,a$ occur once in $\mathcal{A}$ and
  Action~(\ref{algo-msa-not-simple-pair}) replaces them with
  two different variables $U,U' \in \cU$.
  So the loop at line~\ref{algo-msa-while-not-simple} sets
  $\cB$ to $\{p(U,U')\}$ and the algorithm returns $\{X/U,Y/U'\}$.
\end{example}

\end{appendix}


\end{document}